\documentclass[11pt]{article}

\usepackage{stylesheet}
\newcommand{\ENVY}{\mathsf{ENVY}}

\newcommand{\agents}{\mathcal{N}}
\newcommand{\items}{\mathcal{G}}
\newcommand{\alloc}{A}
\newcommand{\xalloc}{X}
\newcommand{\sink}{Sk}
\usepackage{dsfont}

\usepackage{tabularx}
\usepackage{tikz}
\newcolumntype{M}{>{\centering\arraybackslash}m{1cm}}

\newcommand\tikzmark[2]{%
\tikz[remember picture,baseline] \node[inner sep=0.1pt,outer sep=0] (#1){#2};%
}

\newcommand\link[2]{%
\begin{tikzpicture}[remember picture, overlay, >=stealth, shift={(0,0)}]
  \draw[-implies,double equal sign distance] (#1) to (#2);
\end{tikzpicture}%
}

\usepackage{capt-of}
\usepackage{makecell}

\usepackage{footnote}
\makesavenoteenv{algorithm}

\title{Fairness-Efficiency Tradeoffs in Dynamic Fair Division}
\author{David Zeng\thanks{Carnegie Mellon University, davidzeng.42@gmail.com} \and Alexandros Psomas\thanks{Google Research, alexpsomi@gmail.com}}
\date{}

\begin{document}
\maketitle

\begin{abstract}
We investigate the tradeoffs between fairness and efficiency when allocating indivisible items over time. Suppose $T$ items arrive over time and must be allocated upon arrival, immediately and irrevocably, to one of $n$ agents. Agent $i$ assigns a value $v_{it} \in [0,1]$ to the $t$-th item to arrive and has an additive valuation function. If the values are chosen by an adaptive adversary, that gets to see the (random) allocations of items $1$ through $t-1$ before choosing $v_{it}$, it is known that the algorithm that minimizes maximum pairwise envy simply allocates each item uniformly at random; the maximum pairwise envy is then sublinear in $T$, namely $\tilde{O}(\sqrt{T/n})$. If the values are independently and identically drawn from an adversarially chosen distribution $D$, it is also known that, under some mild conditions on $D$, allocating to the agent with the highest value --- a Pareto efficient allocation --- is envy-free with high probability.

In this paper we study fairness-efficiency tradeoffs in this setting and provide matching upper and lower bounds under a spectrum of progressively stronger adversaries. On one hand we show that, even against a non-adaptive adversary, there is no algorithm with sublinear maximum pairwise envy that Pareto dominates the simple algorithm that allocates each item uniformly at random. On the other hand, under a slightly weaker adversary regime where item values are drawn from a known distribution and are independent with respect to time, i.e. $v_{it}$ is independent of $v_{i\hat{t}}$ but possibly correlated with $v_{\hat{i}t}$, optimal (in isolation) efficiency is compatible with optimal (in isolation) fairness. That is, we give an algorithm that is Pareto efficient ex-post and is simultaneously optimal with respect to fairness: for each pair of agents $i$ and $j$, either $i$ envies $j$ by at most one item (a prominent fairness notion), or $i$ does not envy $j$ with high probability. 
En route, we prove a structural (and constructive) result about allocations of divisible items that might be of independent interest: there always exists a Pareto efficient divisible allocation where each agent $i$ either strictly prefers her own bundle to the bundle of agent $j$, or, if she is indifferent, then $i$ and $j$ have identical allocations and the same value (up to multiplicative factors) for all the items that are allocated to them.
\end{abstract}

\section{Introduction}

A set of $T$ indivisible goods has to be allocated to a set of $n$ agents with additive utilities, in a way that is \textit{fair} and \textit{efficient}. A standard fairness concept is envy-freeness, which requires that each agent prefers her own allocation over the allocation of any other agent. Even though envy is clearly unavoidable in this context --- consider the case of a single indivisible good and two agents --- providing \emph{approximately} envy-free solutions is possible~\cite{caragiannis2016unreasonable,lipton2004approximately}. Specifically, an allocation is envy-free up to one item (EF1) if for every pair of agents $i$ and $j$, any envy $i$ has for $j$ can be eliminated by removing at most one good from $j$'s bundle. Recently,~\citet{caragiannis2016unreasonable} show that the allocation that maximizes the product of the agents' utilities (with ties broken based on the number of agents with positive utility) is EF1 and Pareto efficient. %; this algorithm now has real world impact via the non-profit website Spliddit~\cite{goldman2014spliddit}. 

The majority of the literature to date has focused on the case where the items are available to the algorithm upfront. 
In many situations of interest, however, items arrive \emph{online}.
A paradigmatic example is that of food banks~\cite{aleksandrov2015online,lee2019webuildai}. Food banks across the world receive food donations they must allocate; these donations are often perishable, and thus allocation decisions must be made quickly, and donations are typically leftovers, leading to uncertainty about items that will arrive in the future. \citet{BKPP18} study this problem, but focus only on fairness. They show that there exists a deterministic algorithm with \textit{vanishing envy}, that is, the maximum pairwise envy (after all $T$ items have been allocated) is sublinear in $T$, when the value $v_{it}$ of agent $i$ for the $t$-th item is normalized to be in $[0,1]$. Specifically, the envy is guaranteed to be at most $O(\sqrt{T \log T / n})$, and this guarantee is tight up to polylogarithmic factors. The same guarantee can be also achieved by the simple randomized algorithm that allocates each item to a uniformly at random chosen agent. These results hold even against an adaptive adversary that selects the value $v_{it}$ after seeing the allocation of the first $t-1$ items. On the other hand, if we focus only on efficiency, our task is much easier. For example, we could simply allocate each item to the agent with the highest value. But, and this brings us to our interest here, the question remains: \emph{How should we make  allocation decisions online in a way that is fair to the donation recipients, but also as efficient as possible?}

\subsection{Our Contributions}\label{sec:results}

We study the tradeoff between fairness and efficiency in the following setting. $T$ indivisible items arrive online; item $t$ has value $v_{it}$ for agent $i$ and must be allocated immediately and irrevocably. We investigate this tradeoff under a range of adversary models, each one specifying how the values $v_{it}$ are generated. We consider five adversaries and completely characterize the extent to which fairness and efficiency are compatible under each and every one of them.

The weakest adversary we consider simply selects a distribution $D$ from which each value $v_{it}$ is drawn (independently and identically). For this case, even though the setting studied was otherwise unrelated, the answer is essentially given by~\citet{DickersonGKPS14} (and later simplified and improved by~\citet{kurokawa2016can}): under conditions on $D$, the algorithm that allocates each item to the agent with the highest value, an ex-post Pareto efficient algorithm\footnote{See Section~\ref{sec:prelims} for definitions.}, is envy-free with high probability.
Removing the conditions on $D$ is possible: we prove (Proposition~\ref{prop:util-max-properties}) that a simple variation of this algorithm (that remains ex-post Pareto efficient) either outputs an EF1 allocation, or is envy-free with high probability as the number of items goes to infinity. 

Our next, slightly stronger adversary selects a distribution $D_i$ for each agent $i$; agents are independent, but not identical, and items are independent and identical. Unfortunately, as we discuss in Section~\ref{sec:compatible}, it seems very unlikely that the previous greedy approach generalizes even to this adversary. Our first main result is an \emph{optimal} algorithm that works against the even stronger adversary that allows for \emph{correlated} agents, but i.i.d. items: $v_{it}$ can be correlated with $v_{\hat{i}t}$ but not with $v_{i\hat{t}}$. In Theorem~\ref{thm:POCR} we give an ex-post Pareto efficient algorithm that guarantees to every pair of agents $i, j$ that $i$ envies $j$ by at most one item ex-post, or $i$ will not envy $j$ with high probability.

On a high level our approach works as follows. We take the support of the correlated distribution chosen by the adversary\footnote{We assume finite support; see Section~\ref{sec:prelims}.}, scale each item down by the probability it is drawn, and treat this as an offline instance (with $n$ agents and as many items as the support of the joint distribution $D$). We first prove (Lemma~\ref{lem:por-po}) that using a fractional Pareto efficient solution to guide the online decisions (that is, if an $x$ fraction of item $j$ is allocated to agent $i$ in the offline problem, she gets the item with probability $x$ every time it appears) results in \emph{ex-post} Pareto efficient allocations; in fact, applying almost any scaling (or no scaling) doesn't change this. By coupling this with a Pareto efficient and envy-free offline solution, say the fractional allocation that maximizes the product of agents' utilities, we can get an ex-post Pareto efficient algorithm with $\tilde{O}(\sqrt{T})$ envy, since we can treat the pairwise envy between agents $i$ and $j$ contributed by item $t$ as a zero mean random variable with support in $[-1,1]$ (since the $v_{it}$s are in $[0,1]$) and use standard concentration inequalities. If instead the allocation was strongly envy-free, i.e. each agent strictly preferred her own allocation, then these random variables would have negative mean, and the same concentration inequalities would imply negative envy, i.e. envy-freeness, with high probability. This goal is, unfortunately, a bit too optimistic. But, surprisingly, we can provide an allocation with a property (slightly) weaker than strong envy-freeness, which enables, in the online setting, envy-freeness with high probability or EF1 ex-post (the same guarantee as against the weakest adversary!).

Our second main result (Theorem~\ref{thm: structural}) is a structural, and constructive, result about fractional allocations in the offline problem. We give an algorithm that finds a Pareto efficient fractional allocation where agent $i$ either strictly prefers her allocation to the allocation of agent $j$, or, if she is indifferent, then $i$ and $j$ have \emph{identical} allocations and the \emph{same} value (up to multiplicative factors) for all items allocated to them. We do this by giving budgets $\be$ and a solution $(\bx,\bp)$ to the Eisenberg-Gale convex program where $\bx$ has the desired fairness properties. In slightly more detail, we start with the E-G solution where all agents have the same budget, and leverage various properties implied by the KKT conditions in order to iteratively adjust the initial solution and budgets. Our goal is to eliminate edges in the ``indifference graph'', a graph where vertices correspond to agents and there exists a directed edge from $i$ to $j$ if agent $i$ is indifferent between her allocation and the allocation of agent $j$. We end up with an indifference graph that is a disjoint union of cliques, where agents in the same clique have the same allocation and the same value for the items allocated to them (up to multiplicative factors). At this point, the reader familiar with the E-G convex program might be wondering if budget adjustments are even needed.
% That is, why can't we allocate to each agent $1/|S_j|$ of a maximum bang per buck (mbb) item $j$, where $S_j$ is the set of agents who have $j$ as an mbb item at the current prices. Then, if agent $k$ buys something that agent $i$ finds too expensive, $i$ strictly does not envy $k$. Unfortunately, this intuition is misleading, as the resulting allocation might violate, in fact, envy-freeness.
In Appendix~\ref{app: same budget counter example} we provide an example where in order to get the desired structure, giving a different budget to each agent is necessary.
We believe that our result is of independent interest, and our approach might have further applications. 
%We use the final allocation to guide our online decisions. 

Before we proceed to even stronger adversaries, we briefly discuss the chosen fairness criteria. Our results focus on the extent to which efficiency is achievable under the restriction that the algorithm must also guarantee the same fairness criteria that can be achieved in isolation. The fairness guarantee of our main positive result consists of a probabilistic and non-probabilistic part. Even in isolation and under the weakest adversary, this is the best fairness guarantee (with respect to envy) that can be achieved, as we cannot improve on the probabilistic or non-probabilistic guarantee. That is, it is impossible to always output an EF1 allocation (ex-post), and it is also impossible to always output an allocation that is envy-free with high probability. One might further ask if we can give a single non-probabilistic guarantee, but, even envy-freeness up to a constant is not achievable in the online setting; see Appendix~\ref{app:optimal-fairness-guarantee} for these impossibility results. Our algorithms do guarantee vanishing envy, but, while not directly comparable to the ``EF1 or EF w.h.p.'' guarantee, qualitatively this guarantee seems much weaker. Finally, one might attempt to restrict the class of distributions considered. In Appendix~\ref{app:independent agents strong ef} we show that if one excludes point masses, it is possible to guarantee EF with high probability (that is, improve the ``or EF1'' part) for the case of independent agents and i.i.d. items. We do not know if the same result is possible for the case of correlated agents.

Our next (and stronger) adversary is a familiar one: a non-adaptive adversary that selects an instance (with $T$ items) after seeing the algorithm's ``code'', but without knowing the random outcomes of coins used by the algorithm. Of course, as in all the previous cases, the result of~\citet{BKPP18} still applies, so $\tilde{O}(\sqrt{T/n})$ envy is still possible via a deterministic algorithm, but also by simply allocating each item uniformly at random. Our third main result (Theorem~\ref{thm:eff-envy-lb-det}) shows that the latter algorithm is essentially optimal in terms of Pareto efficiency among all algorithms with sublinear envy! Specifically, we show that for every $\varepsilon > 0$, there is no algorithm with sublinear envy that achieves $(1/n + \varepsilon)$-Pareto efficiency ex-ante; see Section~\ref{sec:prelims} for definitions, but, intuitively an allocation is $\alpha$-Pareto efficient, for $\alpha \in (0,1]$, if it is not possible to improve all agents' utilities by a factor of $1/\alpha$. This result extends to our last adversary, the fully adaptive one.

See Table~\ref{table:summary} for a summary of our results.

\renewcommand{\arraystretch}{3}

\begin{table}[ht]
\centering
\resizebox{0.9\columnwidth}{!} {
            \begin{tabular}{c|c|c}
            \multicolumn{1}{c|}{Setting}    & \multicolumn{1}{c|}{Lower Bound}    & \multicolumn{1}{c}{Upper Bound} \\ \cline{1-3}
            \makecell{Identical Agents\\i.i.d. items} 		& \makecell{(EF1 or EF w.h.p.) impossible \\ \tikzmark{a}{to improve (Appendix~\ref{app:optimal-fairness-guarantee})}} & \makecell{(EF1 or EF w.h.p.) and ex-post PO \\ (Prop.~\ref{prop:util-max-properties}, essentially~\citep{DickersonGKPS14})} \\ \cline{1-3}
            \makecell{Independent Agents\\i.i.d. items} 	& \makecell{\tikzmark{b}{(EF1 or EF w.h.p.) impossible} \\ \tikzmark{c}{to improve}}    &  \makecell{\tikzmark{f}{(EF1 or EF w.h.p.) and ex-post PO}} \\ \cline{1-3}
            \makecell{Correlated Agents\\i.i.d. items}		& \makecell{ \tikzmark{d}{(EF1 or EF w.h.p.) impossible} \\ to improve } & \makecell{\tikzmark{e}{(EF1 or EF w.h.p.) and ex-post PO}\\(Theorem~\ref{thm:POCR})} \\ \cline{1-3}
            Non-Adaptive   							& \makecell{Vanishing envy and ex-ante $1/n +\varepsilon$ \\ \tikzmark{g}{Pareto is impossible (Theorem~\ref{thm:eff-envy-lb-det})}}    & \makecell{Vanishing envy and\\ \tikzmark{k}{ex-ante $1/n$ Pareto}} \\ \cline{1-3}
            Adaptive   								& \makecell{\tikzmark{t}{Vanishing envy and ex-ante $1/n +\varepsilon$} \\ Pareto is impossible}    & \makecell{\tikzmark{l}{Vanishing envy and}\\ ex-ante $1/n$ Pareto \cite{BKPP18}} \\ \cline{1-3}
            \end{tabular}
            \link{a}{b}
            \link{c}{d} 
            \link{e}{f} 
            \link{g}{t}
            \link{l}{k}}
            \caption{Comparing the compatibility of efficiency and envy in each setting.}\label{table:summary}
\vspace{-3mm}
\end{table}

\vspace{-3mm}
\subsection{Related Work and Roadmap}

Our paper is related to the growing literature on \emph{online} or \emph{dynamic} fair division (\cite{BKPP18, heachieving, kash2014no, friedman2015dynamic, friedman2017controlled, li2018dynamic, freeman2018dynamic, aleksandrov2015online, walsh2011online, bogomolnaia2019simple}). The work most closely related to ours is by~\citet{BKPP18}. They study envy-freeness in isolation against a fully adaptive adversary, and design a deterministic algorithm (specifically a derandomization of the simple algorithm that allocates each item uniformly at random) that achieves maximum envy of $\tilde{O}(\sqrt{T/n})$. They also show that this is tight up to polylogarithmic factors. In contrast, here we study the tradeoff between fairness and efficiency under different adversaries; against the fully adaptive adversary we prove that there is no algorithm with sublinear (in $T$) envy that outperforms allocating uniformly at random in terms of efficiency. In the same setting, ~\citet{heachieving} study the number of \emph{adjustments} that are necessary and sufficient in order to maintain an EF1 allocation online. 
Very recently,~\citet{bansal2019online} give an algorithm that guarantees envy at most $O(\log T)$ with high probability (efficiency is not considered), for the case of two i.i.d. agents. Notably, as opposed to our positive result here, their result allows for the adversary's distribution to depend on $T$. Even though the setting studied is completely different, the proof that allocating to the agent with the highest value is optimal against the weakest adversary (identical agents and i.i.d. items) is given by~\citet{DickersonGKPS14} (a similar statement also appears in~\citet{kurokawa2016can}), where it is shown that this algorithm outputs an envy-free allocation with probability that goes to $1$ as the number of items goes to infinity; we go over their result in more detail, as well as argue about why it seems unlikely that it generalizes to stronger adversaries, in Section~\ref{sec:compatible}.

For the offline problem, i.e. when all item values are available to the algorithm,~\citet{caragiannis2016unreasonable} show that there is no tradeoff between fairness and efficiency. The (integral) allocation that maximizes the product of the agents' utilities is Pareto efficient and EF1, simultaneously. More recently,~\citet{barman2018finding} show that there always exists an allocation that is EF1 and fractionally Pareto efficient (and also give pseudopolynomial time algorithm for finding allocations that are EF1 and Pareto efficient). 
Computing the fractional allocation that maximizes the product of utilities is a special case of the Fisher market equilibrium with affine utility buyers; the latter problem was solved in (weakly) polynomial time by~\citet{devanur2008market}, improved to a strongly polynomial time algorithm by~\citet{orlin2010improved}. Our structural result starts from an exact solution to the Eisenberg-Gale convex program and then uses a polynomial number of operations. Therefore, all our algorithms run in strongly polynomial time; we further comment on this in Section~\ref{sec: conclusion}.

%A lot of recent work in this space concerns computation, e.g. the complexity of computing a solution to the geometric mean of agents' utilities (whose $n$-th power is the product) (\citet{lee2017apx}), or approximating such a solution (\citet{cole2015approximating,anari2017nash,cole2017convex}, with the current champion being~\citet{barman2018finding} who a $1.45$-approximation algorithm). In this paper we do not focus on such issues, but, assuming an exact solution for the E-G convex program 

It is worth pointing out a connection between what we call the indifference graph and the so-called maximum bang per buck (MBB) graph. In our indifference graph there is an edge from $i$ to $j$ if $i$ is indifferent between her allocation and the allocation of agent $j$. We show (Lemma~\ref{lem:edge-indifferent}) that this condition is equivalent to $v_{ik}/p_k = v_i(\alloc_i)/e_i$, for all items $k$ that $j$ gets a non-zero fraction of, where $p_k$ is the price of item $k$ and $e_i$ is the budget of agent $i$. On the other hand, the MBB graph is bipartite, one side is agents and the other is items, and there is an edge from agent $i$ to item $k$ if $v_{ik}/p_k = v_i(\alloc_i)/e_i$. Properties of MBB graphs have been crucial in recent algorithmic progress on approximating the Nash Social Welfare, e.g. ~\citet{cole2015approximating,garg2018approximating,chaudhury2018fair}, as well as in computing equilibria in Arrow-Debreu
exchange markets (\citet{garg2019strongly}), but beyond the similarity in the definition we are unaware of any technical overlap.

\paragraph{Roadmap.}
Section~\ref{sec:prelims} poses our model and some preliminaries. 
In Section~\ref{sec:compatible} we prove our main positive result: for all i.i.d. distributions over items (with possibly correlated agents), there exists an algorithm that always outputs a Pareto efficient allocation (i.e. this is an ex-post Pareto efficient algorithm) that guarantees that for each pair of agents $i$ and $j$, either $i$ envies $j$ by at most one item, or $i$ does not envy $j$ with high probability. As we mentioned earlier, our algorithm relies on a structural result about divisible allocations; we prove this result in Section~\ref{sec: CISEF}.
In Section~\ref{sec:incompatible} we prove our main negative result: there is a non-adaptive adversary strategy, such that no vanishing envy algorithm has efficiency better than allocating uniformly at random.
We conclude in Section~\ref{sec: conclusion}.

\section{Preliminaries}\label{sec:prelims}

We study the problem of allocating a set of $T$ indivisible items (also referred to as goods) arriving over time, labeled by $\items =  \{ 1, 2, \cdots, T \}$, to a set of $n$ agents, labeled by $\agents = \{ 1, \dots, n \}$. Agent $i \in \agents$ assigns a (normalized) value $v_{it} \in [0,1]$ to each item $t \in \items$. Agents have additive utilities for subsets of items: the value of agent $i$ for a subset of the items $S$ is $v_i(S) = \sum_{t \in S} v_{it}$. An allocation $\alloc$ is a partition of the items into bundles $\alloc_1, \dots, \alloc_n$, where $\alloc_i$ is assigned to agent $i \in \agents$.

Items arrive in order, one per round, for a total of $T$ rounds. The agents' valuations for the $t$-th item become available only after the item arrives, and are unknown until then. Let $\items^t = \{ 1, 2, \dots, t \}$ be the set of items that have arrived up until time $t$. We denote an allocation of $\items^t$ by $\alloc^t$.
We would like to allocate the goods to the agents in a way that the final allocation  $\alloc^T$ is \textit{fair} and \textit{efficient}.
Before we formally define our notions of fairness and efficiency, we go over the different adversary models, each one specifying how the agents' item values are generated.

\paragraph{Adversary Models.}
One can think of each scenario as a game between the adversary and the algorithm. 
For the first three adversaries it will be more intuitive to think of the adversary picking her strategy first, followed by the algorithm, that gets to first see the adversary's strategy. For the last two, it will be more intuitive to think of the algorithm's code being fixed before the adversary picks her strategy. We list our adversaries from weakest to strongest, where a stronger adversary can simulate the strategy of a weaker adversary but not vice versa. All distributions are assumed to be discrete with finite support.
(1) \textbf{Identical agents and i.i.d. items.}  The adversary selects a distribution $D$. On round $t$, the value of item $t$ to each agent $i$ is drawn independently from this distribution, i.e. $v_{it} \sim D$. (2) \textbf{Independent agents and i.i.d. items.}
The adversary selects a distribution $D_i$ for each agent $i$. On round $t$, the value of item $t$ to each agent $i$ is drawn independently from their value distribution, that is $v_{it} \sim D_i$. (3) \textbf{Correlated agents and i.i.d. items.}
The adversary specifies a joint distribution $D$ with marginals $D_1, \ldots, D_n$. On round $t$, the vector of values $v_t = (v_{1t},\dots,v_{nt})$ is drawn i.i.d. from $D$. That is, $v_{it}$ can be correlated with $v_{jt}$, but not with $v_{i\hat{t}}$. For simplicity, we treat this setting as follows. Each item $t$ has one of $m$ types. Agent $i$ has value $v_i(\gamma_j)$ for an item of type $\gamma_j$; the type of each item is drawn i.i.d. from a distribution $D$ with support $G_D$, $|G_D|=m$. We write $f_D(\gamma_j)$ for the probability that the $t$-th item has type $\gamma_j$. Our main positive result is for this setting.
(4) \textbf{Non-adaptive adversary.} The adversary selects an instance (with $n$ agents  and $T$ items) after seeing the algorithm's code, but doesn't know the outcomes of the random coins flipped by the algorithm. Our main negative result is for this setting.
(5) \textbf{Adaptive Adversary} When the adversary selects the value $v_{it}$ she has access to the algorithm's decisions for the first $t-1$ items. This is the setting studied in~\citet{BKPP18}.

\paragraph{Known vs Unknown $T$.} A subtle point that needs to be addressed is whether $T$ is fixed or not when the adversary picks her strategy. For the latter two adversaries we assume that $T$ is fixed (the results of~\cite{BKPP18} also need this assumption). For the first three adversaries it will be convenient to think of the adversary picking a distribution and then study our algorithm's behavior as $T$ goes to infinity. That is, the distribution picked by the adversary does not depend on $T$; so, it cannot for example have support of size $T$ or variance $1/T$, and so on. This allows us to bypass common issues regarding the dependence on the number of items and existence of fair allocations in probabilistic settings (e.g. the answer could depend on whether $T$ is divisible by $n$); see~\cite{DickersonGKPS14,kurokawa2016can,manurangsi2019envy} for some examples where this dependence is crucial. It is worth noting that the recent result of~\citet{bansal2019online} for the case of two correlated agents allows the agents' distribution to depend on $T$ (but only gets an $O(\log T)$ bound on envy, with high probability, and no guarantees about efficiency).
\vspace{-2mm}
\paragraph{Fairness and Efficiency.}
The utility profile of an allocation $\alloc$ is a vector $u = (u_1, \ldots, u_n)$ where $u_i = v_i(\alloc_i)$.
A utility vector $u$ dominates another utility vector $u'$, denoted by $u \succ u'$, if for all $i \in \agents$, $u_i \geq u'_i$ and for some $j \in \agents$, $u_j > u'_j$. An allocation that achieves utility $u$ is \textit{Pareto efficient} if there is no allocation with utility vector $u'$ such that $u' \succ u$. Uncertainty about the future when making allocation decisions will make it impossible to provide efficiency and fairness guarantees simultaneously. In order to measure the efficiency of our algorithms we instead use the following notion of approximate Pareto efficiency, initially defined by~\citet{ruhe1990varepsilon}: an allocation with utility profile $u$ is $\alpha$-Pareto efficient, if there is no feasible utility profile in which, compared to $u$, every agent is (strictly) more than $1/\alpha$ times better off. %A much more demanding variant is the following (\citet{friedman2019fair}): a utility profile $u$ is $\alpha$-strongly Pareto efficient if there is no feasible utility profile in which, compared to $u$, no agent is worse off, and at least one agent is more than $1/\alpha$ times better off. 

Since our setting is online, we also need to further specify whether our guarantees are worst-case or average-case with respect to the adversary instance and the randomness of our algorithms. For a worst-case guarantee, we say that an allocation is $\alpha$-Pareto efficient ex-post if it always outputs an $\alpha$-Pareto efficient allocation (that is, for all agent valuations and all possible outcomes of the random coins flipped by the algorithm). On the other hand, an allocation algorithm is $\alpha$-Pareto efficient ex-ante if the expected allocation is $\alpha$-Pareto efficient (where the expectation is with respect to the randomness in the instance and the algorithm). Our main positive result/algorithm will guarantee $1$-Pareto efficiency ex-post, while our main negative result shows that a certain notion of fairness is incompatible with $1/n$-Pareto efficiency ex-ante.

Regarding fairness, in this paper we focus on envy. Given an allocation $\alloc = (\alloc_1, \ldots, \alloc_n)$, let $\ENVY_{i, j}(\alloc) = \max \{ v_i(\alloc_j) - v_i(\alloc_i), 0 \}$ be the pairwise envy between agents $i$ and $j$, and let $\ENVY(\alloc) = \max_{i,j \in \agents} \ENVY_{i, j}(\alloc)$ be the maximum envy. Clearly, if $\ENVY(\alloc) = 0$ the allocation is envy-free. An allocation $\alloc$ is EF1, if for all pairs of agents $i,j$ $\ENVY_{i, j}(\alloc) \leq \max_{t \in A_j} v_{it}$. Note that this is a stronger guarantee than $\ENVY(\alloc) \leq 1$, since the highest value an agent has can be less than $1$. An algorithm has vanishing envy if the expected maximum pairwise envy is sublinear in $T$, that is $\E[\ENVY(\alloc)] \in o(T)$, or just $\lim_{T \rightarrow \infty} \E[\ENVY(\alloc)]/T \rightarrow 0$.

%\anote{perhaps we should just say exactly what we mean every time it's relevant instead}
%Finally, we say that an algorithm has guarantees a property $P$ with high probability if $P$ occurs with probability at least $1 - O(\frac{1}{T})$. Specifically, one of our main guarantees  will be envy-freeness with high probability, by which we mean that for  

\vspace{-2mm}

\section{When Fairness and Efficiency are Compatible}\label{sec:compatible}

In this section we go over the first three, weaker adversaries, under which fairness and efficiency are compatible.
We start with the easiest setting, identical agents with i.i.d. items, in Section~\ref{sec:util}. We then proceed, in Sections~\ref{subsec:warm up} and~\ref{subsec:mainalgo}, to our main positive result, an algorithm for correlated agents with i.i.d. items that gives the optimal fairness and efficiency guarantees. This, of course, implies the same result for independent agents with i.i.d. items. In Section~\ref{subsec:warm up} we highlight some key insights, while avoiding some of the technical obstacles, and give an ex-post Pareto efficient algorithm with a slightly weaker fairness guarantee. We give our main algorithm in Section~\ref{subsec:mainalgo}.

\subsection{Identical Agents with i.i.d. Items}
\label{sec:util}
Consider an adversary that picks a single distribution $D$, with support $G_D$ of size $m$, and $v_{it}$ is sampled i.i.d. from $D$, for all agents $i$ and all items $t$. Then, efficiency and fairness can be simultaneously achieved using the following algorithm.

\begin{algorithm}
\eIf{$D$ is a point mass}
{Allocate each arriving item in a round-robin manner.}
{For each arriving item $t$, allocate it to the agent $i$ with the maximum value $v_{it}$, breaking ties uniformly at random.}
\caption{Utilitarian}\label{alg:util}
\end{algorithm}%

\begin{proposition}\label{prop:util-max-properties}
Algorithm~\ref{alg:util} outputs an allocation that is always Pareto efficient. Furthermore, for all $\varepsilon > 0$, there exists $T_0 = T_0(\varepsilon)$, such that if $T \geq T_0$, the output allocation satisfies EF1 or is envy-free with probability at least $1-\varepsilon$.
\end{proposition}

This result is essentially given by~\citet{DickersonGKPS14}, in a different context.~\citet{DickersonGKPS14} show that in a \emph{static} setting with $K$ items and $n$ agents, where the value of agent $i$ for item $t$ is drawn from a distribution $D_{i}$, then under mild conditions on the distributions, an envy-free allocation exists with probability $1$ as $K$ goes to infinity.

\begin{theorem}[\citet{DickersonGKPS14}]\label{thm: john et al}
Assume that for all $i,j \in \agents$ and items $t$ the input distributions satisfy (1) $\Pr[ \arg\max_{k \in \agents} v_{kt} = \{ i \} ] = 1/n$, and (2) there exist constants $\mu, \mu^*$ such that $$
0 < \E[ v_{it} | \arg\max_{k \in \agents} v_{kt} = \{ j \} ] \leq \mu < \mu^* \leq \E[ v_{it} | \arg\max_{k \in \agents} v_{kt} = \{ i \} ].$$
Then for all $n \in O( K/ \log K )$, allocating each item to the agent with the highest value is envy-free with probability $1$ as $K \rightarrow \infty$.
\end{theorem}

The first condition implies that each player receives roughly $1/n$ of the goods, and the second condition ensures that each player has higher expected value for each of his own goods compared to goods allocated to other players. Given this result, it is straightforward to prove Proposition~\ref{prop:util-max-properties}. The proof is relegated to Appendix~\ref{app: missing proposition}.

\subsection{Warm-up: Pareto Optimal Rounding}\label{subsec:warm up}

One might ask if we can keep the simplicity of Algorithm~\ref{alg:util} and extend it to work with a stronger adversary, such as when agent's valuations are independent but not necessarily identically distributed. Note that in this case, asking for the probability that agent $i$ has the highest value to be $1/n$ is a fairly strong requirement, so the result of~\citep{DickersonGKPS14} no longer applies. One possible approach is to assign item $t$ to the agent $i$ for whom $F_{D_i}(v_{it})$ is highest, where $F_{D_i}$ is the quantile function for agent $i$'s value distribution. In fact, this approach is fruitful if one focuses solely on fairness, as shown by~\citet{kurokawa2016can}. However, the resulting allocation is not guaranteed to be Pareto efficient. For example, consider a simple scenario with 2 agents: the first agent has values distributed uniformly in $[0,1]$, and the second agent has values uniformly distributed in $[1/2 - \varepsilon, 1/2 + \varepsilon]$ for some small $\varepsilon$. The second agent cares mostly about the number of items she got; items with high quantiles and low quantiles have essentially the same value for her. Therefore, intuitively, exchanging items that agent $2$ won but have large values for agent $1$ (i.e. items where both had high quantiles and $2$ ended up larger) for items that agent $1$ won with low valuation (i.e. items where both had low quantile and $1$ ended up larger) will result in a Pareto improvement.\footnote{For a specific instance, consider the allocation for the following three items: $(0.9, \frac{1}{2} + \varepsilon), (0.1, \frac{1}{2} - \varepsilon), (0.1, \frac{1}{2} - \varepsilon)$. The second agent receives the first item, while the first agent receives the latter two, which is not Pareto optimal.}

All in all, achieving fairness and efficiency simultaneously, beyond identical agents, seems a lot more intricate than fairness or efficiency in isolation. We will skip the independent agent case altogether, and directly study the harder problem of correlated agents: each item $t$ draws its type $\gamma_j$ from a distribution $D$. Items are i.i.d. but agent values can be correlated. Before we present the optimal algorithm we illustrate some key ideas by giving a simple algorithm that achieves ex-post Pareto efficiency and a weaker notion of fairness, namely vanishing envy with high probability. Recall that $f_D(\gamma_j)$ is the probability that an item drawn from $D$ has type $\gamma_j$, and $G_D$, with $|G_D|=m$, is the support of $D$. $v_{i}(\gamma_j)$ is the value of an item of type $\gamma_j$ to agent $i$; note that this is different than $v_{it}$, the value of agent $i$ for the $t$-th item to arrive.

Although the original problem concerns indivisible items, our approach solves an offline divisible item allocation problem as an intermediate step. Rather than each item being assigned to an agent, a fractional allocation is an $n$ by $m$ matrix $X$, where $n$ is the number of agents and $m$ is the support of $D$. For each agent $i$ and item $j$, $X_{ij} \in [0,1]$ represents the proportion of item $j$ allocated to agent $i$. $X$ is constrained so that $\sum_{i \in \agents} X_{ij} = 1$ for all items $j$. The $i^{th}$ row of $X$, $X_i$, can be thought of as the fractional allocation of agent $i$. Finally, we assume agents have linear preferences, where for any two agents $i, j$, $v_i(X_j) = \sum_{k \in [m]} v_{ik} X_{jk}$.

\begin{algorithm}
\leavevmode \newline Input: Distribution $D$ over item types, agent valuation functions $v_i$.
\begin{enumerate}
\item Define $v'_i$ for each agent $i$ as follows. For each $\gamma_j \in G_D$, set $v'_i(\gamma_j) = v_i(\gamma_j) f_D(\gamma_j)$.
\item Find the divisible allocation $X$ of $G_D$ that maximizes the product of utilities w.r.t $v'$.
\item In the online setting, allocate the newly arrived item $t$ with type $\gamma_j$ to agent $i$ with probability $X_{ij}$, for all $t = 1 \dots T$.
\end{enumerate}
\caption{Pareto Optimal Rounding}\label{alg:POR}
\end{algorithm}

We first show that Algorithm~\ref{alg:POR} always produces a Pareto optimal allocation. In fact, we show something much stronger: allocating type $\gamma_j$ to agent $i$ with probability $X_{ij}$ always results in a Pareto efficient allocation (ex-post) for \emph{every} Pareto efficient fractional allocation $X$.

\begin{lemma}
\label{lem:por-po}
Given a distribution $D$ over $m$ item types and valuation function $v_i$ for each agent $i$, let $X$ be a fractional and Pareto optimal allocation of $G_D$ under valuation functions $v'_i$, where $v'_i (\gamma_j) = v_i(\gamma_j) \cdot f_D(\gamma_j)$. Let $S$ be a set of $T$ items, where the type of the $t$-th item is drawn from $D$. Let $\alloc = (\alloc_1, \ldots, \alloc_n)$ be any allocation of $S$ where an item with type $\gamma_j$ is allocated to agent $i$ only if $X_{ij} > 0$.
Then $\alloc$ is Pareto optimal under $v$.
\end{lemma}

%\dnote{I think we still need the first step of the proof so adding it back. However, the first step is now simplified because there's no $\tilde{v}$ (maybe if we want to save space, we can just assert that is true as the proof is fairly straightforward)}

\begin{proof} 
We prove the theorem in two steps. First, we show that $X$ is also Pareto optimal under $v$. Second, we show that if $X$ is Pareto optimal under $v$, then $\alloc$ will be Pareto optimal under $v$.

First, suppose that $X$ is not Pareto optimal under $v$. Let $X'$ be an allocation that dominates $X$ under $v$ and let $\Delta = X' - X$. We construct $\Delta'$, where $\Delta'_{ij} = \Delta_{ij} \cdot \frac{1}{f_D(\gamma_j)}$, so that the change in utilities induced by $\Delta'$ under $v'$ is the same as the change in utilities induced by $\Delta$ under $v$. Therefore, the (possibly infeasible) allocation $X + \Delta'$ dominates $X$ under $v'$. However, for all $c \in [0,1]$, allocation $X + c\Delta$ is feasible and still dominates $X$ under $v$. Setting $c = \min_k f_D(\gamma_k)$, $X + c\Delta'$ is a feasible allocation that Pareto dominates $X$ under $v'$, which contradicts the Pareto optimality of $X$ under $v'$.

Next, let $\mathcal{A}$ be the set of all allocations $A$ (of $T$ items) such that for all $i \in \agents$ and $j \in S$ an item with type $\gamma_j$ is allocated to agent $i$ only if $X_{ij} > 0$. We can write $X$ as $\sum_{A \in \mathcal{A}} p_A \cdot A$, such that $p_A > 0$ for all $A$, and $\sum_{A \in \mathcal{A}} p_A = 1$. This can be done by simply giving each item $j \in S$ to agent $i$ with probability $X_{ij}$. It is easy to confirm that the set of all possible allocations that can be outputs by this procedure is precisely $\mathcal{A}$, and the probability that a specific allocation $A$ is generated is exactly $p_A = \prod_{j \in S} \prod_{i \in \agents} X_{ij} \cdot \mathds{1}\{ \text{$i$ gets item $j$ in $A$} \}$. Finally, since $X = \sum_{A \in \mathcal{A}} p_A \cdot A$, if an allocation $\hat{A} \in \mathcal{A}$ is Pareto dominated by an allocation $\hat{A}'$ under $v$, by replacing the $p_{\hat{A}} \hat{A}$ term by a $p_{\hat{A}} \hat{A}'$ term we get a fractional allocation $X'$ that dominates $X$, contradicting the Pareto optimality of $X$ under $v$.
\end{proof}

Maximizing the product of utilities leads to a Pareto efficient divisible item allocation. Therefore, Lemma~\ref{lem:por-po} implies that Algorithm~\ref{alg:POR} is ex-post Pareto efficient. Algorithm~\ref{alg:POR} guarantees a notion of fairness slightly weaker than vanishing envy: vanishing envy with high probability.
\begin{theorem}
\label{thm:por-ve}
For all $\varepsilon > 0$, there exists $T_0 = T_0(\varepsilon)$, such that if $T \geq T_0$, Algorithm~\ref{alg:POR} outputs an allocation $\alloc$ such that for all agents $i$, $j$, $\ENVY_{i, j}(\alloc) \in o(T)$ with probability at least $1-\varepsilon$.
\end{theorem}

\begin{proof}
The divisible allocation $X$ that maximizes the product of utilities is envy-free~\cite{varian1974equity}.
Thus, we have $\sum_{k \in [m]} v_i(\gamma_k)f_D(\gamma_k)X_{ik} \geq \sum_{k \in [m]} v_i(\gamma_k)f_D(\gamma_k)X_{jk}$ for all pairs of agents $i, j$.
The value of agent $i$ for the bundle allocated to agent $j$ by Algorithm~\ref{alg:POR}, $v_i(\alloc_j)$, is a random variable depending on randomness in both the algorithm and item draws. Let $I^{k, j}_t$ be an indicator random variable for the event that item $t$ is of type $\gamma_k$ and is assigned to agent $j$. We have for any pair of agents $i, j$: $v_i(\alloc_j) = \sum_{t \in [T]} \sum_{k \in [m]} v_i(\gamma_k) I^{k, j}_t$. Therefore, $\E[v_i(\alloc_j)] = T \cdot \sum_{k \in [m]} v_i(\gamma_k)f_D(\gamma_k)X_{jk}$. From the previous discussion, we then have that $\E[v_i(A_i)] \geq \E[v_i(A_j)]$.
Using Hoeffding's inequality with parameter $\delta = \sqrt{T\log{T}}$ we get: $\Pr[v_i(\alloc_i) - \E[v_i(\alloc_i)]] \leq -\sqrt{T\log{T}}] \leq 2 \exp\left(-\frac{2T\log{T}}{T}\right) = \frac{2}{T^2},$
and similarly for the deviation of $v_i(A_j)$. We conclude that by picking $T_0 = \sqrt{4/\varepsilon}$, then with probability at least $1 - \frac{4}{T^2}\geq 1-\varepsilon$, $\ENVY_{i, j}(\alloc) = \max \{ v_i(\alloc_j) - v_i(\alloc_i), 0 \} \leq 2\sqrt{T \log T} \in o(T)$.
\end{proof}

\subsection{Beyond Vanishing Envy: Optimal Fairness for Correlated Agents}
\label{subsec:mainalgo}

In the proof of Theorem~\ref{thm:por-ve} we use standard tail inequalities to show that, with high probability, the envy between any two agents does not deviate from its expectation by more than $O(\sqrt{T \log{T}})$. Since the divisible allocation we are using to round guarantees envy-freeness, this leads to vanishing envy. If instead we were able to find an allocation $X$ for the divisible item problem that guarantees \textit{strong envy-freeness}, where for every pair of agents $i, j$, $v_i(X_i) > v_i(X_j)$, and used $X$ to guide the online decisions, by a similar calculation the final allocation would be envy-free with high probability, since $\E[v_i(\alloc_i)] - \E[v_i(\alloc_j)]$ decreases much faster than the deviation bound of $\sqrt{T \log {T}}$. Unfortunately, strong envy-free allocations do not always exist, even for divisible items; consider the case of two agents with identical valuations. Interestingly, in isolation this condition is also sufficient: if no two players have identical valuation functions (up to multiplicative factor), there exists a strongly envy-free allocation; see~\citet{barbanel2005geometry}.
However, if we want both Pareto optimality and strong envy-freeness, the condition is no longer sufficient; see Appendix~\ref{app:sef-po-nonident-cex} for an example\footnote{\citet{barbanel2005geometry} gives a sufficient condition for Pareto optimality and strong envy-freeness to be simultaneously achievable. Unfortunately, this condition cannot be satisfied here: it roughly asks for every pair of agents to have a different value for \emph{every} subset (allowing for fractional items) of items. See Theorem 12.36 in \citet{barbanel2005geometry} for the exact statement and proof.}.

Nevertheless, we can still achieve a notion of fairness that is weaker than strong envy-freeness, but still sufficient for our purposes. We say that agent $i$ is indifferent to agent $j$ if $v_i(\xalloc_i) = v_i(\xalloc_j)$. We then view indifference as a directed graph, where the nodes consist of the $n$ agents and there is an edge from $i$ to $j$ if $i$ is indifferent to $j$. For a divisible item allocation $\xalloc$, let $I(\xalloc)$ refer to the indifference graph. Note that for an envy-free allocation $\xalloc$, lack of an edge from node $i$ to node $j$ in $I(\xalloc)$ implies that $v_i(\xalloc_i) > v_i(\xalloc_j)$, i.e. strong envy-freeness between the two agents.

\begin{definition}\label{def:cisef}
A divisible allocation $\xalloc$ is \textit{clique identical strongly envy-free} (CISEF) if $\xalloc$ is envy-free and there exists a partition of agents $\agents$ into $s$ sets $C_1, \ldots, C_s$ such that: (1) For each $i \in [s]$, $C_i$ is a clique in $I(\xalloc)$, (2) for each $i, j \in [s]$ where $i \neq j$, there are no edges between $C_i$ and $C_j$, (3) for each $i \in [s]$, for each pair of agents $j, k \in C_i$, $X_j = X_k$, and (4) there exists positive scalars $r_1, \ldots, r_n$ such that for each $i \in [s]$, pair of agents $j, k \in C_i$ and item $l$ where $X_{jl} > 0$ (or $X_{kl} > 0$), $\frac{v_{jl}}{r_j} = \frac{v_{kl}}{r_k}$.
\end{definition}
\vspace{-2mm}
In other words, an allocation is CISEF if $\xalloc$ is envy-free, the graph $I(\xalloc)$ is a disjoint union of cliques, agents in a clique have identical allocations, and agents have identical valuations (up to a multiplicative factor) over items allocated to any agent in the clique.

The intuition behind this definition is the following. As discussed in Appendix~\ref{app:sef-po-nonident-cex}, it is impossible to find a strongly envy-free and Pareto optimal allocation in situations where some agents have valuations that are identical for the items that they could receive in a Pareto optimal allocation. In the definition of CISEF, these too similar agents will be given identical allocations and end up in the same clique. However, in all other cases, the allocation is strongly envy-free.

In Theorem~\ref{thm:main-disjoint-cliques} in Section~\ref{sec: CISEF} we prove our main structural result. We will show that for any divisible item problem instance, we can find an allocation $\xalloc^*$ that is both Pareto efficient and CISEF. In the remainder of this section we present a slightly modified version of Algorithm~\ref{alg:POR} that, given such a Pareto efficient and CISEF divisible allocation, remains Pareto efficient ex-post and also achieves the target fairness properties.

\begin{algorithm}
Input: Item distribution $D$, agent valuation functions $v_i$.
\begin{enumerate}
\item Define $v'_i$ for each agent $i$ as follows. For each $\gamma_j \in G_D$, set $v'_i(\gamma_j) = v_i(\gamma_j) f_D(\gamma_j)$.
\item Compute a divisible allocation $X^*$ of $G_D$ under valuation functions $v'_i$ that is both Pareto efficient and CISEF (see Section~\ref{sec: CISEF} for an algorithm). Let $C_1, \cdots, C_s$ be the disjoint cliques of the graph.
\item In the online setting, assign the newly arrived item $t$ with type $\gamma_j$ to clique $C_i$ with probability $\sum_{k \in C_i} X^*_{kj}$. After an item is assigned to a clique $C_i$, allocate it among the agents in $C_i$ by giving it to the agent in $C_i$ who has received the least value so far, according to \emph{all} agents in the clique.\footnote{$i$ thinks that $argmin_{j \in \agents} v_i(\alloc^t_j)$ has the smallest value so far. We select the agent with the smallest value according to all $i \in \agents$, which is unique, up to tie breaking, since all agents agree on the value of all items that have gone to the clique (up to multiplicative factors).}
\end{enumerate}
\caption{Pareto Optimal Clique Rounding}\label{alg:POCR}
\end{algorithm}

\begin{theorem}\label{thm:POCR}
Algorithm~\ref{alg:POCR} is ex-post Pareto efficient and, for all $\varepsilon > 0$, there exists $T_0 = T_0(\varepsilon)$, such that if $T \geq T_0$, for all agents $i, j$, one of the two guarantees holds: either $i$ envies $j$ by at most one item with probability $1$ or $i$ will not envy $j$ with probability at least $1-\varepsilon$.
\end{theorem}

\begin{proof}
Let $\alloc$ be the allocation produced by Algorithm~\ref{alg:POCR} and $X^*$ be the fractional allocation found for the divisible item allocation problem. Pareto efficiency is immediately implied by Lemma~\ref{lem:por-po}, since $X^*$ is Pareto efficient and an item of type $\gamma_j$ is allocated to agent $i$ only if $X^*_{ij} > 0$.

For any two agents $i, j$, there are two cases. In the first case, $i$ and $j$ belong to the same clique $C_k$. Let $S$ be the set of items assigned to $C_k$ during the execution of Algorithm~\ref{alg:POCR}, i.e. $S = \cup_{\ell \in C_k} \alloc_\ell$. 
Agents in $C_k$ have identical valuations up to a multiplicative factor for the items that they get with positive probability. Therefore, giving each item to the agent that has received the least value so far (according to any agent, as they rank allocations of $S$ in the same order) ensures that $\ENVY_{i,j}(\alloc)$ is at most the maximum value that $i$ has for any item in $S$.

Now consider the case that $i$ and $j$ belong to different cliques, $C_i$ and $C_j$ respectively. By the definition of a CISEF allocation, we know that $v_i(X^*_i) = v_i(X^*_j) + c$ for some constant $c>0$. Let $\tilde{\alloc}_i$ be the fractional allocation where agent $i$ receives a $1/\abs{C_i}$ fraction of the items assigned to her clique, i.e. $\tilde{\alloc}_{it} = \frac{1}{\abs{C_i}} \mathds{1} \{ t \in A_k \text{ for some agent $k \in C_i$} \}$, and similarly for $\tilde{\alloc}_j$. Since agents in a clique receive identical allocations in $\alloc$ we have that $\E[v_i(\tilde{\alloc}_i) - v_i(\tilde{\alloc}_j)] = T\E[v_i(X^*_i) - v_i(X^*_j)] = cT$.

We can apply Hoeffding's inequality to show that with probability at least $1-\Theta(1/T^2) \geq 1-\varepsilon$, $v_i(\tilde{\alloc}_j) - v_i(\tilde{\alloc}_i) < 2\sqrt{T \log{T}} - cT$, which is negative for sufficiently large $T$. Specifically, we can pick $T$ large enough so that $v_i(\tilde{\alloc}_j) - v_i(\tilde{\alloc}_i) < -2$ with high probability. Finally, observe that $\abs{v_i(\alloc_i) - v_i(\tilde{\alloc}_i)} \leq 1$. This is because $\tilde{\alloc}_i$ is the average allocation of agents in the clique, and, as we argued earlier, the maximum envy for two agents in the same clique is at most $1$, thus $v_i(\alloc_j) - v_i(\alloc_i) < v_i(\tilde{\alloc}_j) - v_i(\tilde{\alloc}_i) + 2$. Combined with the deviation bound, we conclude that $v_i(\alloc_j) - v_i(\alloc_i) \leq 0$ with probability at least $1-\varepsilon$.
\end{proof}

\section{Achieving Pareto Efficiency and CISEF}\label{sec: CISEF}

In this section we prove our main structural result.

\begin{theorem}\label{thm: structural}
Given any instance with $m$ divisible items and $n$ additive agents, there always exists an allocation that is simultaneously clique identical strongly envy-free (CISEF) and Pareto efficient.
\end{theorem}

%{\color{red} As discussed in the previous section, there does not seem to be a simple condition that guarantees even the existence of a SEF and PO allocation. Thus, we propose CISEF and PO as a relaxation of SEF and PO whose existence can be guaranteed for all instances.}

We prove Theorem~\ref{thm: structural} by building on a standard approach for finding an envy-free and Pareto efficient allocation, namely solving the Eisenberg-Gale convex program (henceforth E-G program). Recall that the E-G program with ``budgets'' $\be$ is the following.
\begin{align*}
    \max & \quad \sum_{i = 1}^n e_i \log \sum_{j = 1}^m v_{ij}x_{ij}, &\text{subject to} & \quad \sum_{i=1}^n x_{ij} \leq 1, \forall j \in [m], \text{ and }\quad x_{ij} \geq 0, \forall i \in [n], j \in [m].
\end{align*}

When $e_i=e_j$ for all $i, j$, the outcome is also known as  the competitive equilibrium from equal incomes (CEEI), which is envy-free \cite{varian1974equity}. There exists a solution to the E-G program with primal variables $\bx$ and dual variables $\bp$ (dual variable $p_j \geq 0$ corresponds to the first constraint above) satisfying the following conditions.
\begin{align}
        &\textstyle \forall j \in [m] : p_j > 0 \implies \sum_{i=1}^n x_{ij} = 1 \label{eq: KKT 1} \\
        &\textstyle \forall i \in [n], j \in [m] : \frac{v_{ij}}{p_j} \leq \frac{\sum_{k=1}^m v_{ik}x_{ik}}{e_i} \label{eq: KKT 2}\\
        &\textstyle \forall i \in [n], j \in [m] : x_{ij} > 0 \implies \frac{v_{ij}}{p_j} = \frac{\sum_{k=1}^m v_{ik}x_{ik}}{e_i} = \max_{k \in [m]} \frac{v_{ik}}{p_k}.\label{eq: KKT 3}
\end{align}

These conditions are both necessary and sufficient for a feasible solution to be optimal, and can be derived from the KKT conditions; for completeness we show the derivation in Appendix~\ref{app:kkt-derive}. The standard interpretation is that $e_i$ is the budget of agent $i$ and $p_j$ is the price for item $j$. Then, a solution $\bx, \bp$ consists of prices $\bp$ and allocations $\bx$, such that each agent spends their entire budget $e_i$ on ``optimal'' items and all items are completely sold. We say that an item $j$ is ``optimal'' for agent $i$ given prices $\bp$, when it maximizes the ratio $v_{ij}/p_j$, also known as the \emph{bang-per-buck}. 

For the remainder of this section, given a solution $\bx, \bp$, we use $\bx$ and $\xalloc$ indistinguishably for the allocation, we write $\xalloc_i$ for the allocation of agent $i$, and we say that item $k$ is allocated to $i$, $k \in \xalloc_i$, if $x_{ik} > 0$.
We assume without loss of generality that for any solution, $\bx, \bp$, we have $\forall j : p_j > 0$. This holds as long as each item has at least one agent who values it; if this is not the case we can safely drop those items.
We prove the following, which immediately implies Theorem~\ref{thm: structural}\footnote{Pareto efficiency is implied since the objective function is monotone.}.

\begin{theorem}
\label{thm:main-disjoint-cliques}
There exist budgets $\be$ and an optimal solution $(\bx = \xalloc, \bp)$ to the E-G convex program with budgets $\be$, such that $\xalloc$ is clique identical strongly envy-free.
\end{theorem}

\noindent We start with the optimal solution $\xalloc$ to the E-G convex program with identical budgets $e_i = 1$ for all $i$. Then, at a high level, we break the algorithm into two procedures that jointly alter $\bx, \bp$ and $\be$ such that $\bx, \bp$ remains an optimal solution to the convex program with budgets $\be$, while preserving envy-freeness, until $\xalloc$ satisfies the desired properties. Specifically, the indifference graph $I(\xalloc)$ will end up being a disjoint set of cliques, such that agents in a clique have identical allocations.

\subsection*{Optimal Transfers}

Given an allocation, let $r_i := \frac{v_i(A_i)}{e_i}$ be the maximum bang-per-buck ratio of agent $i$. We say that agent $i$ is indifferent towards any item $k$ for which $\frac{v_{ik}}{p_k} = r_i$. We first give a useful property of solutions $\bx, \bp$. 
The proof is relegated to Appendix~\ref{app: missing lemma proof}.

\begin{lemma}
\label{lem:edge-indifferent}
Given a solution $\bx, \bp$ of an E-G program with budgets $\be$, for all agents $i, j$ such that $e_i = e_j$, $v_i(\xalloc_i) = v_i(\xalloc_j)$ if and only if $\forall k \in \xalloc_j : \frac{v_{ik}}{p_k} = r_i$.
\end{lemma}

Lemma~\ref{lem:edge-indifferent} tells us that under equal budgets, if agent $i$ is indifferent to agent $j$'s allocation, then $j$'s items are maximum bang-per-buck items for $i$ as well. This gives some intuition for our approach. Assuming equal budgets, if we move items along indifference edges, we can avoid violating the KKT conditions and the solution remains optimal. We formalize this idea in Lemma~\ref{lem:optimal-transfers}; first we need the following definition. Given a solution $(\bx = \xalloc, \bp)$ for budgets $\be$, we consider a change in allocation of items $\Delta$, where $\Delta_{ik}$ is the difference in the allocation of item $k$ for agent $i$. 

\begin{definition}[Optimal transfer]
A transfer of items $\Delta$ is an \textit{optimal transfer} if for all items $k$ (1) $\sum_{i \in [n]} \Delta_{ik} = 0$, i.e. the total allocation of item $k$ remains unchanged, (2) $x_{ik} + \Delta_{ik} \in [0,1]$ for all agents $i$, i.e. $\bx + \Delta$ is feasible, and (3) for all agents $i$ such that $\Delta_{ik} > 0$, $\frac{v_{ik}}{p_k} = r_i$, i.e. if agent $i$ is given more of item $k$, then item $k$ maximizes bang-per-buck for agent $i$. %Finally, $\xalloc + \Delta$ should be a feasible allocation.
\end{definition}

The proof of the next lemma can be found in Appendix~\ref{app: lemma optimal transfers}.

\begin{lemma}
    \label{lem:optimal-transfers}
    Let $(\bx = \xalloc, \bp)$ be a solution for budgets $\be$ and $\Delta$ be an optimal transfer. Let $\delta$ represent the change in budget where $\delta_i = \sum_{k \in \items} p_k \cdot \Delta_{ik}$. Let $\xalloc' = \xalloc + \Delta$. Then $(\bx' = A', \bp)$ is a solution for budgets $\be' = \be + \delta$.
\end{lemma}

\subsection*{Indifference Edge Elimination}

For an allocation $\xalloc$ and subset of agents $S \subseteq \agents$, we overload notation, and let $\xalloc_S$ refer to the allocation for agents in $S$. A set of agents $S$ have identical budgets under $\be$ if there exists some $c$ where for all $i \in S$, $e_i = c$.
Recall that for an allocation $\xalloc$,  $I(\xalloc)$ refers to the indifference graph, a graph where we have a vertex for every agent and an edge from (the vertices corresponding to) agent $i$ to agent $j$ if $v_i(\xalloc_i) = v_i(\xalloc_j)$. In the remainder of this section we refer to agents and vertices interchangeably. Also, recall that for a directed graph $G = (V,E)$, a clique is a subset of vertices $S \subseteq V$ such that for all $v \in S, u \in S$, where $u \neq v$, there is an edge $(u,v) \in E$, and that a weakly connected component (henceforth just component) $S$ is a subset of the agents such that for each pair of agents $i, j \in S$, there is either a path from $i$ to $j$ or a path from $j$ to $i$, and $S$ is a maximal subgraph with this property.

\begin{definition}[Clique acyclic graph]\label{dfn: clique acyclic}
A directed graph $G=(V,E)$ is clique acyclic if the vertices can be partitioned into cliques $C_1, \ldots, C_k$, that is, $C_i \subseteq V$ for all $i$, $\cup_{i=1}^k C_i = V$, and $C_i \cap C_j = \emptyset$ for all $i \neq j$, and  where for any cycle $K$ in the graph, $K$ only contains vertices from $C_i$ for some $i$.
\end{definition}

A crucial step towards producing a CISEF allocation will be to find an allocation $\xalloc$ such that $I(\xalloc)$ is clique acyclic and envy-free, where agents in each clique have the same allocation.

\begin{lemma}
\label{lem:const-util-shifts}
There exists an algorithm that takes as input a solution $(\bx = \xalloc, \bp)$ for budgets $\be$ and a component $S$ with identical budgets, and finds a solution $(\bx' = \xalloc', \bp)$ for budgets $\be$ where $I(\xalloc'_S)$ is clique acyclic and agents in each clique have identical allocations without violating envy-freeness or adding new indifference edges to $I(\xalloc)$. 
\end{lemma}

Take the allocation $\xalloc$ and indifference graph $I(\xalloc_S)$. At a high level we attempt to apply one of the following two operations; we only apply Operation $2$ only if Operation $1$ cannot be done. 
\begin{itemize}[leftmargin=*]
    \item \textbf{Operation 1:} Eliminate every cycle that is not a clique.
    \item \textbf{Operation 2:} Partition the graph into cliques by merging cliques and ``re-balancing'' allocations.
\end{itemize}

We explain in detail how these operations work, and prove that they satisfy some basic properties. We proceed to show how to use them to prove the lemma.

\paragraph{Operation $1$.}
Without loss of generality, suppose we have a cycle $K = (1, \ldots, k)$, where edges go from agent $i$ to agent $i + 1$ (and indices wrap around). If there exists an $i$ where there is no edge from $i - 1$ to $i + 1$, we show how to eliminate at least one edge from the cycle, keep $v_i(\xalloc_i)$ the same for all agents, and not create any new indifference edges. %; later in this proof we use this operation repeatedly to eliminate non-clique cycles.

To perform this operation notice that Lemma~\ref{lem:edge-indifferent} implies that there exists an item $\ell \in \xalloc_{i + 1}$ that is not a bang-per-buck item for agent $i-1$, i.e. $\frac{v_{i-1,\ell}}{p_\ell} < r_{i - 1}$. We construct an optimal transfer $\Delta$, parameterized by a budget $b$, as follows. Agent $i$ will take $b$ worth of item $\ell$ (specifically $b/p_\ell$ units) from agent $i + 1$. All other agents $i' \in K$, $i' \neq i$, will take arbitrary items of total worth $b$ from $i' + 1$.   Crucially, for all of the transfers, since there is an $(i',i'+1)$ edge for all $i' \in K$, i.e. $v_{i'}(\xalloc_{i'}) = v_{i'}(\xalloc_{i'+1})$, Lemma~\ref{lem:edge-indifferent} implies the items taken maximize bang-per-buck for $i'$. Next, notice that we can find a $b > 0$ small enough that $\xalloc + \Delta$ is a feasible allocation\footnote{Letting $e_S$ be the budget of agents in $S$, we can choose any $b \leq \min(e_S, x_{i+1, l} \cdot p_l) = x_{i+1, l} \cdot p_l$ (i.e we just need to ensure agent $i$ does not take more than $x_{i+1, l}$ of item $l$).}. Finally, our transfers preserve the total allocation of each item and thus, $\Delta$ is an optimal transfer. Each agent loses and gains $b$ worth of items so there is no change in budget associated with $\Delta$. Therefore, Lemma~\ref{lem:optimal-transfers} implies that $(\bx' = \xalloc', \bp)$ is a solution for budgets $\be$. Thus for all agents $i$, $v_i(\xalloc_i) = v_i(\xalloc'_i)$. 

Next, observe that $v_{i - 1}(\xalloc'_i) < v_{i - 1}(\xalloc_i)$ since it decreases by $b \cdot r_{i - 1}$ but increases by strictly less than $b \cdot r_{i - 1}$. Since $v_{i - 1}(\xalloc_{i - 1}) = v_{i - 1}(\xalloc'_{i - 1})$, this implies that the indifference edge from agent $i - 1$ to agent $i$ disappears. Finally, we want to ensure that we do not violate envy-freeness or add new indifference edges. There are two cases. The first case is that $i$ is indifferent to $j$ in $\xalloc$: we only modify allocations in $S$, and $S$ is a component, so we only consider $i \in S, j \in S$. Because $(\bx', \bp)$ is a solution for budgets $\be$, and agents in $S$ have identical budgets, agent $i$ will still not envy $j$ in $\xalloc'$. The second case is that $i$ is not indifferent to $j$ in $\xalloc$; we want to find a $b > 0$ small enough such that $i$ will not envy or be indifferent to $j$ in $\xalloc'$. We give such a $b$ in Appendix~\ref{app:choosing b}.

We can repeatedly use the above process to eliminate all non-clique cycles. We do this by eliminating cycles in order of increasing size. Suppose that all cycles of size $k$ form a clique. It is then possible to ensure that for all cycles of size $k + 1$, the vertices form a clique. To see why, consider an arbitrary cycle $K$ of size $k + 1$. If there exists an $i$ such that there is no edge from agent $i - 1$ to agent $i + 1$, we can eliminate the cycle. Otherwise, for all $i$, there is an edge from agent $i - 1$ to agent $i + 1$. Then, we know that for all $i$, $K \setminus \{ i \}$ is a cycle of length $k$, and therefore a clique of size $k$, implying the vertices of $K$ form a clique of size $k + 1$. Note that any cycle of size $2$ immediately forms a clique. Therefore, if we repeatedly choose the smallest size non-clique cycle and eliminate it, we will eventually eliminate all cycles that are not a clique. 

\paragraph{Operation $2$.}

We construct a set of cliques $C_1, \ldots, C_s$ by starting with each vertex in its own clique and arbitrarily merging cliques if the resulting set of vertices still forms a clique.  Suppose we merge to get a clique $C = \{ 1, \ldots, \ell \}$. Lemma~\ref{lem:edge-indifferent} implies that every agent $i \in C$ is indifferent to any item $z \in \xalloc_C$, $\xalloc_C = \xalloc_1 + \ldots + \xalloc_{\ell}$, such that $\frac{v_{iz}}{p_z} = r_i$. Thus, we can perform the following re-balance operation to form $\xalloc'$: for each agent $i \in C$, set $\xalloc'_i = \frac{1}{\abs{C}}\xalloc_C$, and for all other agents $i \notin C$, $\xalloc'_i = \xalloc_i$. 

First, note that this is an optimal transfer, by definition. We want to show that we do not violate envy-freeness or add new indifference edges. We separate into cases based on whether agents $i, j$ are in $C$. If both agents are in $C$, nothing changes. If $i \in C$, $j \notin C$, since $v_i(\xalloc_i)$ stays the same, nothing changes. If $i \notin C$, $j \in C$, if $i$ was indifferent to all agents in $C$, nothing changes. Otherwise, $i$ is not indifferent to some agent in $C$, in which case after the re-balance, $i$ will lose their indifference edge towards all agents in $C$. 

\begin{proof}[Proof of Lemma~\ref{lem:const-util-shifts}]
We start by applying Operation $1$, which eliminates every cycle that is not a clique. Now, notice that eliminating all non-clique cycles does \emph{not} imply that we can partition the graph into cliques in a way that is clique acyclic. To see this most clearly, consider a $5$ vertex instance, where vertices $1$, $2$, $3$ form a clique, and so do vertices $3$, $4$, $5$. All cycles are also cliques, but the graph is still not clique acyclic (the most ``tempting'' partition has the issue that cycle $( 1, 2, 3 )$ contains a vertex that belongs to two cliques). This is where Operation $2$ comes in.

Once we have eliminated all non-clique cycles, if we are not in a clique acyclic graph with the desired properties, we apply Operation $2$. It is possible that during the execution of Operation $2$, an indifference edge will be eliminated. In this case, we stop with Operation $2$ and go back to applying Operation $1$, and so on. Eventually, this process terminates: neither operation creates any new edges, and furthermore, each time we loop we eliminate at least one edge. The two operations preserve the property that $(\bx' = \xalloc', \bp)$ is a solution for budgets $\be$, and that $\xalloc'$ is envy-free, without adding new indifference edges. In addition, re-balancing in Operation $2$ ensures that agents in each clique have identical allocations. It remains to show that $I(\xalloc')$ is clique acyclic.

Assume for sake of contradiction that there exists a cycle that includes vertices in two different cliques $C_1, C_2$. Therefore there exists some edge from $C_1$ to $C_2$ and an edge from $C_2$ to $C_1$. Due to the fact that agents in a clique have identical allocations, this implies that there exists agents $i_1 \in C_1, i_2 \in C_2$ such that $i_1$ has edges to all of $C_2$ and $i_2$ has edges to all of $C_1$. Thus, we can construct a cycle, and thus a clique, containing all agents in $C_1 \cup C_2$. Then, $C_1 \cup C_2$ form a clique, contradicting the fact that no more mergers were possible by Operation $2$.
\end{proof}

Once we have an allocation with the properties of Lemma~\ref{lem:const-util-shifts}, it remains to eliminate the edges between cliques, while preserving the property that agents in a clique have the same allocation.

\begin{lemma}\label{lem:budget-change-shifts}
There exists an algorithm that takes as input an allocation $\xalloc$ and component $S$ with identical budgets such that $I(\xalloc_S)$ is clique acyclic (with at least one non-clique edge) and agents in each clique have identical allocations, and finds a set of agent budgets $\be'$ and solution $(\bx' =\xalloc', \bp)$ for budgets $\be'$, where $I(\xalloc'_S)$ consists of $k > 2$ components $S_1, S_2, \cdots, S_k$, each component consists of agents with identical budgets, $I(\xalloc'_S)$ has strictly fewer edges than $I(\xalloc_S)$, and the new allocation does not violate envy-freeness or add new indifference edges.
\end{lemma}

\begin{proof}
For $I(\xalloc_S)$, we view each clique as a vertex. Consider the graph $G$, where each vertex is a clique $C_i$ (of $I(\xalloc_S)$), and there is an edge between $C_i$ and $C_j$ if there exists a $v_i \in C_i, v_j \in C_j$ where $(v_i, v_j)$ is an edge in $I(\xalloc_S)$. Observe that $G$ forms a DAG, and let $C_s$ be a source vertex in $G$ and let $\sink = \{ C_1, \ldots, C_l \}$ be the sink vertices reachable from $C_s$. We know a source and sink vertex exists because we assume there is at least one non-clique edge.

Now we return to $I(\xalloc_S)$. We will create item transfers $\Delta$. We find it intuitive to describe $\Delta$ via a flow in the following graph. Starting from $G$, create a new global source vertex $s$, and add all $(s,v)$ edges for all $v \in C_s$. Create a new global sink vertex $t$, and add a $(v,t)$ edge for all $v \in \sink$. Finally, we let each edge in the graph have infinite capacity. Next, we find a flow of size $b$ from $s$ to $t$, with the additional constraints that the flow from $s$ to each $v \in C_s$ is the same, and for each $C \in \sink$, the flow from each $v \in C$ to $t$ is the same. We show how to construct such a flow, starting from an arbitrary flow of size $b$. Let $f_{ij}$ be the flow along edge $(i, j)$. For the source clique, $C_s$, we can update $f'_{si} = \frac{1}{\abs{C_s}}\sum_{i \in C_s} f_{si}$ so that $f'$ now satisfies the equal flow constraint for the source clique. Then, we can always ensure the flow $f'$ is balanced by updating the flow between vertices in $C_s$: For each $i, j \in C_s, i < j$, set $f'_{ij} = f_{ij} + \frac{1}{\abs{C_s}} \left(f_{sj} - f_{si}\right)$, where negative flows are added as positive flows in the reverse direction. The flow along all other edges remains the same. To show that $f'$ is a feasible flow, we show flows are balanced for vertices in the source clique. Observe for an agent $i \in C_s$:
\[\sum_{j \in \agents \cup \{ s \}} f'_{ji} - \sum_{j \in \agents} f'_{ij} = \sum_{j \in \agents} f_{ji} + \frac{1}{\abs{C_s}}\sum_{j \in C_s} f_{sj} + \frac{1}{\abs{C_s}} \sum_{j \in C_s} \left(f_{si} - f_{sj}\right) - \sum_{j \in \agents} f_{ij} = \sum_{j \in \agents \cup \{ s \}} f_{ji} - \sum_{j \in \agents} f_{ij}.\]
Therefore, if $f$ is a valid flow, so is $f'$. An analogous procedure can be applied to the sink cliques.

We use this flow to guide the item transfers, $\Delta$. For each edge with flow, agent $i$ will take an arbitrary $f_{ij}$ worth of items from agent $j$. The global sink and source vertices are, obviously, excluded. We choose $b$ small enough such that $\xalloc + \Delta$ is feasible. We can assume without loss of generality that the sum of flows towards any vertex is at most $b$, as any extra flow must be part of a cycle and can be eliminated. Then if $e_S$ is the budget of agents in $S$, we can choose any $b \leq e_S$. Since we ensure the total allocation of an item is preserved and we only transfer items along indifference edges between agents of the same budget, $\Delta$ is an optimal transfer.  Therefore we can apply Lemma~\ref{lem:optimal-transfers}.

Since there are no item transfers to $s$, each agent $v \in C_s$ with an incoming flow has the same increased budget. Similarly, for each sink clique $C$ with a positive flow from $v \in C$ to $t$, each member of $C$ will have the same decreased budget. 
Now, if we take each component in the resulting indifference graph, we claim that each agent in the component will have identical budgets.
It is sufficient to show that if $e'_i \neq e'_j$, then there will not be an edge from $v_i$ to $v_j$. Since initially everyone in $S$ had the same budget, and item transfers --- as well as ``budgets'' --- go from sinks to sources, if $e'_i < e'_j$, there was never an edge from $v_i$ to $v_j$ to begin with. Now suppose $e'_i > e'_j$; we have:
\[v_i(\xalloc_i) = \sum_{k \in \xalloc_i} v_{ik} x_{ik} =^{KKT~\eqref{eq: KKT 3}} \sum_{k \in \xalloc_i} r_i p_k x_{ik} = r_i e'_i > r_i e'_j = \sum_{k \in \xalloc_j} r_i p_k x_{jk} \geq \sum_{k \in \xalloc_j} v_{ik} x_{jk} = v_i(\xalloc_j),\]
which implies no $(i,j)$ indifference edge.
Since agents in $C_s$ have a greater budget than all other remaining  agents, there will be at least two components.
The final property that we need to prove is that the new allocation does not violate envy-freeness or add new indifference edges. We break into cases. First, if agents $i$ and $j$ were indifferent, as we only modify allocations in component $S$, we only consider when $i \in S, j \in S$. We ensure $v_i(\xalloc_i)$ either stays the same or increases, since we transfer items in the opposite direction of indifference edges. For similar reasons, $v_i(\xalloc_j)$ either stays the same or decreases. On the other hand, if agent $i$ did not envy $j$ we can set $b$ small enough such that $i$ will not envy $j$ in $\xalloc'$. We discuss this final detail in Appendix~\ref{app:choosing b}.
\end{proof}
\subsection*{Putting everything together}

%We are now ready to prove Theorem~\ref{thm:main-disjoint-cliques}.

%\mainDisjointCliques*

\begin{proof}[Proof of Theorem~\ref{thm:main-disjoint-cliques}]
Our overall algorithm first solves the E-G convex program with budgets $\be$ where $e_i = 1$ for all $i$, to find a solution $(\bx = \xalloc, \bp)$. We keep track of the set of components $\calS$ with identical budgets, where initially $\calS = \{ \agents \}$. We alternate between applying the algorithm of Lemma~\ref{lem:const-util-shifts} and Lemma~\ref{lem:budget-change-shifts}, henceforth procedure $1$ and procedure $2$. We start by applying the former to each $S \in \calS$. It is possible that after applying procedure $1$, edges are eliminated, resulting in $S$ being split into multiple components. Let $f(S)$ be the set of components formed and update $\calS := \bigcup_{S \in \calS} f(S)$. The ``clique acylic'' and ``identical allocations'' properties are still satisfied by each individual component in $\calS$. We then apply procedure $2$ to each $S \in \calS$. We perform the same update to $\calS$, where $f(S)$ is the set of components with identical budgets found after applying procedure $2$. We repeat until applying the two procedures does not decrease the number of edges in the graph. Finally, we perform the re-balance operation (described in procedure $1$) on each component.% (or clique, as we will show).

Both procedures do not add edges and reduce the number of edges. There are at most $n^2 - n$ initial edges, and thus the algorithm terminates. In addition, both procedures produce allocations and budgets where $(\bx' = A', \bp)$ is a solution to the E-G program for budgets $\be'$. Finally, $\xalloc$ is CISEF. First, $\xalloc$ is envy-free by construction. Next, consider the final graph $I(\xalloc)$. We claim that each component is a clique where agents have identical budgets. The component must be clique-acyclic after procedure $1$ is applied. No edges were removed by procedure $2$, so there could not have been any edges between cliques. Therefore, the component can only consist of one clique. Next, procedure $2$ guarantees that each agent in the clique has the same budget. Lemma~\ref{lem:edge-indifferent} tells us that agents will have identical valuations, up to a multiplicative factor, for items that any agent in the clique receives.  The final re-balance operation doesn't alter any of the above properties.
\end{proof}

% \subsection*{Simplifying Fairness Guarantee}

% While the main positive result in Theorem~\ref{thm:POCR} applies to both correlated and independent agents, the fairness guarantee achievable for independent agents can be characterized further. When agents have non-identical but independent value distributions, under the mild assumption that none of the distributions are point masses, we give an ex-post Pareto efficient algorithm that guarantees envy-freeness with high probability (just like the result for identical agents). The algorithm, described in Appendix~\ref{app:independent agents strong ef}, builds on the approach for correlated agents, but utilizes the special structure of possible item types. 
\section{When Fairness and Efficiency are Incompatible}\label{sec:incompatible}

In this section we discuss the tradeoff between efficiency and fairness against the latter two, stronger adversaries. Our main conclusion is, informally, that no allocation algorithm with vanishing envy Pareto dominates random allocation.~\citet{BKPP18} show that against an adaptive adversary (and thus against all weaker adversaries), random allocation has vanishing envy. Allocating at random is also $1/n$-Pareto efficient ex-ante.
Here, we show that no algorithm can achieve vanishing envy and expected $(\frac{1}{n} + \varepsilon)$-Pareto efficient for any $\varepsilon > 0$, even against the non-adaptive adversary.
\begin{theorem}
\label{thm:eff-envy-lb-det}
For any $\varepsilon > 0$, no randomized allocation algorithm can achieve both expected envy $f(T)$, where $f(T) \in o(T)$, and $(\frac{1}{n} + \varepsilon)$-Pareto efficient ex-ante against a non-adaptive adversary.
\end{theorem}

\begin{proof}
To build up intuition, we first describe the lower bound for the case of an adaptive adversary. Since against an adaptive adversary randomization does not help, we focus on deterministic algorithms. Consider any vanishing envy algorithm that for all $T$ produces an allocation $\alloc^T$, where $\ENVY(\alloc^T) \leq f(T)$ for some $f(T) \in o(T)$, and assume, for the sake of contradiction, that this algorithm achieves $(\frac{1}{n} + \varepsilon)$-Pareto efficiency for some $\varepsilon > 0$. 

Consider the following instance $I$, parameterized by $\varepsilon$ and $T$. For each agent $i$, items $j$ from  $\frac{T}{n}(i - 1)$ to $\frac{T}{n}i$ will have value $1$, and all other items $j'$ have value $ \varepsilon$. For all intermediate allocations at time $t \leq T$, we must have $\ENVY(\alloc^t) \leq f(T)$ since an adaptive adversary can make the remaining items valueless to all agents. We start by showing via induction that for all ``segments'' of items $\frac{T}{n}(i - 1)$ to $\frac{T}{n}i$, each agent must receive a number of items in $[\frac{T}{n^2} - x_i, \frac{T}{n^2} + x_i]$, where $x_i = \frac{f(T)}{\varepsilon} \left(1 + \frac{2}{\varepsilon}\right)^{i - 1}$.

For $i=1$, i.e. the first segment, suppose that some agent $k$ receives $\frac{T}{n^2} + y$ items where $y > 0$. Another agent $\hat{k}$ must then receive fewer than $\frac{T}{n^2}$ items. Then, the envy of $\hat{k}$ for $k$ at the end of the first segment, $\ENVY_{\hat{k}, k}(\alloc^{T/n})$ is at least $\varepsilon \cdot y$. But, $\ENVY_{\hat{k}, k}(\alloc^{T/n}) \leq f(T)$, which implies that $y \leq \frac{f(T)}{\varepsilon}$; the lower bound on $y$ is identical. For the inductive step, again suppose that in the segment $\frac{T}{n}(i - 1)$ to $\frac{T}{n}i$ some agent $k$ receives $\frac{T}{n^2} + y$ items, where $y > 0$, and let $\hat{k}$ be the agent who received fewer than $\frac{T}{n^2}$ items. At the start of the segment, to agent $\hat{k}$, the difference in value between agent $k$'s bundle and their bundle is at least $-2 \sum_{i' < i} x_{i'}$, which is if $\hat{k}$ got $\frac{T}{n^2} + x_{i'}$ in each segment, $k$ got $\frac{T}{n^2} - x_{i'}$, and $\hat{k}$ had value $1$ for all items up until $\frac{T}{n}(i - 1)$. Thus, after the $i$-th segment, $\ENVY_{\hat{k}, k}(\alloc^{\frac{T}{n}i}) \geq \varepsilon \cdot y + v_{\hat{k}}\left(\alloc^{\frac{T}{n}(i-1)}_{k}\right) - v_{\hat{k}}\left(\alloc^{\frac{T}{n}(i-1)}_{\hat{k}}\right) \geq \varepsilon \cdot y - 2 \sum_{i' < i} x_{i'},$ which in turn implies that
\[ \textstyle y \leq \frac{1}{\varepsilon}\left(f(T) + 2\sum_{i' < i}x_{i'}\right) = \frac{1}{\varepsilon}\left(f(T) + 2 \sum_{i' < i} \frac{f(T)}{\varepsilon} \left(1 + \frac{2}{\varepsilon}\right)^{i' - 1} \right) = \frac{f(T)}{\varepsilon} \left(1 + \frac{2}{\varepsilon} \right)^{i-1}.
\]
The last equality is obtained by summing the geometric series and simplifying. The bound for $y$ is identical when we consider the case that $y < 0$.
Next, we show that this implies the allocation $\alloc^T$ is not $(\frac{1}{n} + \varepsilon)$-Pareto efficient. First, note that the social welfare maximizing allocation achieves utility $(\frac{T}{n}, \ldots, \frac{T}{n})$, by giving all the items of the $i$-th segment to agent $i$. Meanwhile, since $x_i < x_n$, we have that in $\alloc^T$ each agent gets utility $u_i$ at most $(1 + (n - 1)\varepsilon)(\frac{T}{n^2} + x_n)$. Therefore,
\begin{align*}
 \frac{u_i}{1/n + \varepsilon} &< (1 + (n - 1)\varepsilon)\left(\frac{T}{n^2} + x_n\right)\left(\frac{1}{\frac{1}{n} + \varepsilon}\right)= (1 + (n - 1)\varepsilon)\left(\frac{T}{n^2} + \frac{f(T)}{\varepsilon} (1 + \frac{2}{\varepsilon})^{n - 1}\right)\frac{n}{1 + \varepsilon n} \\
   &= \frac{1 + (n - 1)\varepsilon}{1 + \varepsilon n} \cdot \left(\frac{T}{n} + n\frac{f(T)}{\varepsilon}(1 + \frac{2}{\varepsilon})^{n - 1}\right) =\frac{T}{n} \cdot \frac{1 + (n - 1)\varepsilon}{1 + \varepsilon n} \cdot \left(1 + \frac{n^2}{\varepsilon} \frac{f(T)}{T}(1 + \frac{2}{\varepsilon})^{n - 1}\right).
\end{align*}
Picking $T$ large enough so that $\frac{f(T)}{T} < \frac{\varepsilon}{1+(n-1)\varepsilon} \cdot \frac{\varepsilon}{n^2 (1 + 2/ \varepsilon)^{n - 1} }$ gives $u_i < \frac{T}{n} \cdot ( 1/n + \varepsilon )$ for every agent $i$. Therefore, $\alloc^T$ is not $(\frac{1}{n} + \varepsilon)$-Pareto efficient, a contradiction. This concludes the proof for an adaptive adversary.
We can use the above result to prove the same for a non-adaptive adversary; details are relegated to Appendix~\ref{app: lower bound proof continued}.
\end{proof}
\section{Conclusion and Future Directions}\label{sec: conclusion}

In this paper we study whether fairness and efficiency can be balanced in an online setting under a spectrum of different adversary models. We show tight upper and lower bounds for each setting. For the adaptive adversary, and therefore for all the other, weaker ones as well, the upper bound is given by~\cite{BKPP18}: an algorithm with vanishing envy that gets a $1/n$ approximation of Pareto. We show that this result is essentially tight (Theorem~\ref{thm:eff-envy-lb-det}) for the non-adaptive adversary (and therefore for all the stronger ones) if one wants to balance fairness and efficiency. Our main upper bound (Theorem~\ref{thm:POCR}) is for correlated agents and i.i.d. items and shows how to get Pareto efficiency ex-post and guarantee EF1 ex-post, or envy-freeness with high probability. We argue (Appendix~\ref{app:optimal-fairness-guarantee}) that one cannot hope for a better fairness guarantee even for identical agents, and since the efficiency guarantee is the strongest one, we conclude that this result is also tight. Finally, even though our main upper bound of course holds for identical agents, for this case we show (Proposition~\ref{prop:util-max-properties}) how to adapt the approach of~\citet{DickersonGKPS14} and~\citet{kurokawa2016can} and get the same tight result, but with a simpler algorithm.

\paragraph*{Computation Considerations.}

Given our structural result, Theorem~\ref{thm: structural}, it is clear that all our results run in polynomial time. To get our structural result we assume an exact solution to the E-G program, that we can get in strongly polynomial time via the results of~\citet{orlin2010improved}. Our edge-elimination processes only runs $O(n^2)$ times. The only possible issue is the number of bits in the solution $(\bx, \bp)$ and budgets $\be$, as the item transfers described in Lemma~\ref{lem:const-util-shifts} and~\ref{lem:budget-change-shifts} can both increase the length (in bits) of $\bx$ and $\be$. This increase depends primarily on the budget transfers $b$ (computed in Appendix~\ref{app:choosing b}). We can always choose $b$ such that $b$ is equal to $(v_i(\alloc_i) - v_i(\alloc_j))/4c$ for some $i, j$, where $c$ is a constant that only depends on the instance and the initial solution to the E-G program. Since $v_i(A_i), v_i(A_j)$ are linear functions of $\bx$, their representations are a constant (depending on the $v_{ij}$ and $n$) larger than the elements of $\bx$. In addition, the representation of $b$ is an additive constant larger than the bit length of the min difference $v_i(\alloc_i)-  v_i(\alloc_j)$, so performing the transfer of items under budget $b$ will also only increase the bit length of the elements of $\bx, \be$ by a constant.

\paragraph*{Future directions.}
Notice that, with the exception of the non-adaptive adversary, we completely understand what is possible if we want to optimize only fairness or only efficiency. On the other hand, for the non-adaptive adversary, it remains open how well one can do with respect to minimizing the maximum (expected) pairwise envy. Vanishing envy is of course achievable, but we do not know if it's possible to improve the $\tilde{O}(\sqrt{T})$ guarantee, or how big that improvement could be; for instance, we don't even know of a super constant lower bound. Another technical question that remains open is what is possible when the distributions chosen by the adversary can depend on $T$.
Finally, a general direction for future work is balancing fairness and efficiency under different adversaries (for example, one could imagine adversaries weaker than the non-adaptive one, but stronger than correlated agents with i.i.d. items), or under different definitions of ``fair'' and ``efficient''.

%

%Our optimization problem above, specifically the computation of the allocation which maximizes the product of utilities, involved real valued quantities. Generally, approximately maximizing the product of utilities is not guaranteed to be envy-free...

%\dnote{TODO: Discuss computation here, such as precision issues when solving EG convex program.}

%\subsubsection*{Future directions}

%\begin{itemize}
%    \item Non-IID item (so somewhere in between distribution over instances and IID items. e.x markov chain?)
%    \item How to allocate items online when the items are i.i.d but the distributions are not known
%\end{itemize}

\bibliographystyle{plainnat}
\bibliography{references}

\begin{thebibliography}{27}
\providecommand{\natexlab}[1]{#1}
\providecommand{\url}[1]{\texttt{#1}}
\expandafter\ifx\csname urlstyle\endcsname\relax
  \providecommand{\doi}[1]{doi: #1}\else
  \providecommand{\doi}{doi: \begingroup \urlstyle{rm}\Url}\fi

\bibitem[Aleksandrov et~al.(2015)Aleksandrov, Aziz, Gaspers, and
  Walsh]{aleksandrov2015online}
Martin~Damyanov Aleksandrov, Haris Aziz, Serge Gaspers, and Toby Walsh.
\newblock Online fair division: analysing a food bank problem.
\newblock In \emph{Twenty-Fourth International Joint Conference on Artificial
  Intelligence}, 2015.

\bibitem[Bansal et~al.(2019)Bansal, Jiang, Singla, and Sinha]{bansal2019online}
Nikhil Bansal, Haotian Jiang, Sahil Singla, and Makrand Sinha.
\newblock Online vector balancing and geometric discrepancy.
\newblock \emph{arXiv preprint arXiv:1912.03350}, 2019.

\bibitem[Barbanel(2005)]{barbanel2005geometry}
Julius~B Barbanel.
\newblock \emph{The geometry of efficient fair division}.
\newblock Cambridge University Press, 2005.

\bibitem[Barman et~al.(2018)Barman, Krishnamurthy, and
  Vaish]{barman2018finding}
Siddharth Barman, Sanath~Kumar Krishnamurthy, and Rohit Vaish.
\newblock Finding fair and efficient allocations.
\newblock In \emph{Proceedings of the 2018 ACM Conference on Economics and
  Computation}, pages 557--574. ACM, 2018.

\bibitem[Benad\`e et~al.(2018)Benad\`e, Kazachkov, Procaccia, and
  Psomas]{BKPP18}
G.~Benad\`e, A.~M. Kazachkov, A.~D. Procaccia, and C.-A. Psomas.
\newblock How to make envy vanish over time.
\newblock In \emph{Proceedings of the 2018 {ACM} Conference on Economics and
  Computation}, pages 593--610, 2018.

\bibitem[Bogomolnaia et~al.(2019)Bogomolnaia, Moulin, and
  Sandomirskiy]{bogomolnaia2019simple}
Anna Bogomolnaia, Herve Moulin, and Fedor Sandomirskiy.
\newblock A simple online fair division problem.
\newblock \emph{arXiv preprint arXiv:1903.10361}, 2019.

\bibitem[Caragiannis et~al.(2016)Caragiannis, Kurokawa, Moulin, Procaccia,
  Shah, and Wang]{caragiannis2016unreasonable}
Ioannis Caragiannis, David Kurokawa, Herv{\'e} Moulin, Ariel~D Procaccia,
  Nisarg Shah, and Junxing Wang.
\newblock The unreasonable fairness of maximum nash welfare.
\newblock In \emph{Proceedings of the 2016 ACM Conference on Economics and
  Computation}, pages 305--322. ACM, 2016.

\bibitem[Chaudhury et~al.(2018)Chaudhury, Cheung, Garg, Garg, Hoefer, and
  Mehlhorn]{chaudhury2018fair}
Bhaskar~Ray Chaudhury, Yun~Kuen Cheung, Jugal Garg, Naveen Garg, Martin Hoefer,
  and Kurt Mehlhorn.
\newblock On fair division for indivisible items.
\newblock In \emph{38th IARCS Annual Conference on Foundations of Software
  Technology and Theoretical Computer Science (FSTTCS 2018)}. Schloss
  Dagstuhl-Leibniz-Zentrum fuer Informatik, 2018.

\bibitem[Cole and Gkatzelis(2015)]{cole2015approximating}
Richard Cole and Vasilis Gkatzelis.
\newblock Approximating the nash social welfare with indivisible items.
\newblock In \emph{Proceedings of the forty-seventh annual ACM symposium on
  Theory of computing}, pages 371--380. ACM, 2015.

\bibitem[Devanur et~al.(2008)Devanur, Papadimitriou, Saberi, and
  Vazirani]{devanur2008market}
Nikhil~R Devanur, Christos~H Papadimitriou, Amin Saberi, and Vijay~V Vazirani.
\newblock Market equilibrium via a primal--dual algorithm for a convex program.
\newblock \emph{Journal of the ACM (JACM)}, 55\penalty0 (5):\penalty0 22, 2008.

\bibitem[Dickerson et~al.(2014)Dickerson, Goldman, Karp, Procaccia, and
  Sandholm]{DickersonGKPS14}
John~P. Dickerson, Jonathan~R. Goldman, Jeremy Karp, Ariel~D. Procaccia, and
  Tuomas Sandholm.
\newblock The computational rise and fall of fairness.
\newblock In \emph{Proceedings of the Twenty-Eighth {AAAI} Conference on
  Artificial Intelligence}, pages 1405--1411, 2014.

\bibitem[Freeman et~al.(2018)Freeman, Zahedi, Conitzer, and
  Lee]{freeman2018dynamic}
Rupert Freeman, Seyed~Majid Zahedi, Vincent Conitzer, and Benjamin~C Lee.
\newblock Dynamic proportional sharing: A game-theoretic approach.
\newblock \emph{Proceedings of the ACM on Measurement and Analysis of Computing
  Systems}, 2\penalty0 (1):\penalty0 3, 2018.

\bibitem[Friedman et~al.(2015)Friedman, Psomas, and Vardi]{friedman2015dynamic}
Eric Friedman, Christos-Alexandros Psomas, and Shai Vardi.
\newblock Dynamic fair division with minimal disruptions.
\newblock In \emph{Proceedings of the sixteenth ACM conference on Economics and
  Computation}, pages 697--713. ACM, 2015.

\bibitem[Friedman et~al.(2017)Friedman, Psomas, and
  Vardi]{friedman2017controlled}
Eric Friedman, Christos-Alexandros Psomas, and Shai Vardi.
\newblock Controlled dynamic fair division.
\newblock In \emph{Proceedings of the 2017 ACM Conference on Economics and
  Computation}, pages 461--478. ACM, 2017.

\bibitem[Garg and V{\'e}gh(2019)]{garg2019strongly}
Jugal Garg and L{\'a}szl{\'o}~A V{\'e}gh.
\newblock A strongly polynomial algorithm for linear exchange markets.
\newblock In \emph{Proceedings of the 51st Annual ACM SIGACT Symposium on
  Theory of Computing}, pages 54--65. ACM, 2019.

\bibitem[Garg et~al.(2018)Garg, Hoefer, and Mehlhorn]{garg2018approximating}
Jugal Garg, Martin Hoefer, and Kurt Mehlhorn.
\newblock Approximating the nash social welfare with budget-additive
  valuations.
\newblock In \emph{Proceedings of the Twenty-Ninth Annual ACM-SIAM Symposium on
  Discrete Algorithms}, pages 2326--2340. SIAM, 2018.

\bibitem[He et~al.(2019)He, Procaccia, Psomas, and Zeng]{heachieving}
Jiafan He, Ariel~D Procaccia, Alexandros Psomas, and David Zeng.
\newblock Achieving a fairer future by changing the past.
\newblock In \emph{Proceedings of the 28th International Joint Conference on
  Artificial Intelligence}, pages 343--349. AAAI Press, 2019.

\bibitem[Kash et~al.(2014)Kash, Procaccia, and Shah]{kash2014no}
Ian Kash, Ariel~D Procaccia, and Nisarg Shah.
\newblock No agent left behind: Dynamic fair division of multiple resources.
\newblock \emph{Journal of Artificial Intelligence Research}, 51:\penalty0
  579--603, 2014.

\bibitem[Kurokawa et~al.(2016)Kurokawa, Procaccia, and Wang]{kurokawa2016can}
David Kurokawa, Ariel~D Procaccia, and Junxing Wang.
\newblock When can the maximin share guarantee be guaranteed?
\newblock In \emph{Thirtieth AAAI Conference on Artificial Intelligence}, 2016.

\bibitem[Lee et~al.(2019)Lee, Kusbit, Kahng, Kim, Yuan, Chan, See, Noothigattu,
  Lee, Psomas, et~al.]{lee2019webuildai}
Min~Kyung Lee, Daniel Kusbit, Anson Kahng, Ji~Tae Kim, Xinran Yuan, Allissa
  Chan, Daniel See, Ritesh Noothigattu, Siheon Lee, Alexandros Psomas, et~al.
\newblock Webuildai: Participatory framework for algorithmic governance.
\newblock \emph{Proceedings of the ACM on Human-Computer Interaction},
  3\penalty0 (CSCW):\penalty0 1--35, 2019.

\bibitem[Li et~al.(2018)Li, Li, and Li]{li2018dynamic}
Bo~Li, Wenyang Li, and Yingkai Li.
\newblock Dynamic fair division problem with general valuations.
\newblock In \emph{Proceedings of the 27th International Joint Conference on
  Artificial Intelligence}, pages 375--381, 2018.

\bibitem[Lipton et~al.(2004)Lipton, Markakis, Mossel, and
  Saberi]{lipton2004approximately}
Richard~J Lipton, Evangelos Markakis, Elchanan Mossel, and Amin Saberi.
\newblock On approximately fair allocations of indivisible goods.
\newblock In \emph{Proceedings of the 5th ACM conference on Electronic
  commerce}, pages 125--131. ACM, 2004.

\bibitem[Manurangsi and Suksompong(2019)]{manurangsi2019envy}
Pasin Manurangsi and Warut Suksompong.
\newblock When do envy-free allocations exist?
\newblock In \emph{Proceedings of the AAAI Conference on Artificial
  Intelligence}, volume~33, pages 2109--2116, 2019.

\bibitem[Orlin(2010)]{orlin2010improved}
James~B Orlin.
\newblock Improved algorithms for computing fisher's market clearing prices:
  computing fisher's market clearing prices.
\newblock In \emph{Proceedings of the forty-second ACM symposium on Theory of
  computing}, pages 291--300. ACM, 2010.

\bibitem[Ruhe and Fruhwirth(1990)]{ruhe1990varepsilon}
G{\"u}nther Ruhe and Bernd Fruhwirth.
\newblock $\varepsilon$-optimality for bicriteria programs and its application
  to minimum cost flows.
\newblock \emph{Computing}, 44\penalty0 (1):\penalty0 21--34, 1990.

\bibitem[Varian(1974)]{varian1974equity}
Hal~R Varian.
\newblock Equity, envy, and efficiency.
\newblock \emph{Journal of Economic Theory}, 9\penalty0 (1):\penalty0 63--91,
  1974.

\bibitem[Walsh(2011)]{walsh2011online}
Toby Walsh.
\newblock Online cake cutting.
\newblock In \emph{International Conference on Algorithmic DecisionTheory},
  pages 292--305. Springer, 2011.

\end{thebibliography}

\appendix
\newpage
\section{Different Budgets and Prices are Necessary}\label{app: same budget counter example}

Consider the instance with $3$ agents and $2$ items. Agent $A$ has value $1$ for both items. Agent $B$ has values $1/2$ and $1$ for the first and second item, respectively, while agent $C$ has values $1$ and $1/2$. 

It is easy to check that the unique solution that maximizes the product of agents' utilities (i.e. equal budgets in the EG convex program) gives agent $A$ $1/3$ of each item, while agent $B$ gets $2/3$ of item 2 and agent $C$ gets $2/3$ of item $1$. The prices of both items are equal.

% Therefore, no matter what these prices are, as long as they're equal, agent $B$ and agent $C$ must get only items $2$ and $1$, respectively. The proposed algorithm of assigning each agent an equal fraction of their mbb items results in a solution that is not envy-free: If agent $A$ spends some of her budget on item $1$, then since we demand the same allocation for every agent that buys item $1$, both agent $A$ and agent $C$ must get half of item $1$. No matter how the second item is divided it is now impossible to guarantee utility $1/3$ for each agent.

In this example, the unique product maximizing solution does not satisfy the CISEF property, as there are two indifference edges from $A$ to $B$ and $A$ to $C$. Therefore, to get an allocation like the one guaranteed by Theorem~\ref{thm: structural} one needs to go beyond equal budgets in the E-G convex program.

\section{Improving on ``EF1 or EF w.h.p.'' Guarantee is Impossible}
\label{app:optimal-fairness-guarantee}
We show these impossibility results under the weakest adversary, where the adversary can give a distribution $D^V$ and item values $v_{it}$ are drawn from $D^V$. 

First, we show that it is not always possible to guarantee envy-freeness with high probability. Define $D^V$ to be the uniform distribution over the set $\{ 1 \}$. Note that whenever $T$ is not a multiple of $n$, the allocation will not be envy-free. Next, we show that for any $x$, envy-freeness up to $x$ goods, is not an achievable guarantee. We use the construction of the lower bound in \citet{BKPP18}, which assumes item values bounded within $[0, 1]$. 

\begin{lemma}[\citet{BKPP18}]\label{thm: benade et al}
    For $n \geq 2$ and $r < 1$, there exists an adversary strategy for setting item values such that any algorithm must have $\ENVY(A) \in \Omega((T/n)^{r/2})$, where $A$ is the allocation $T$ items. 
\end{lemma}

In the proof of Lemma~\ref{thm: benade et al}, they use an adversary strategy where all agents other than the first two agents do not value the item. The value of the arriving item to the first two agents depends solely on a state machine, where the states are ``$\ldots, L_2, L_1, 0, R_1, R_2, \ldots$''. The item associated with state $L_i$ has value $(1, \nu_i)$ (so $a_1$ values it at $1$ and $a_2$ values it at $\nu_i$), the item associated with state $0$ has value $(1, 1)$, and the item associated with state $R_i$ has value $(\nu_i, 1)$. The state machine starts at state $0$ and transitions one step left or right after each item arrival, depending on how the previous item was allocated. Thus, the set of all possible item values used for a length $T$ instance is $\{0, 1, \nu_1, \ldots, \nu_T\}$. They define $\nu_i = (i + 1)^r - i^r$. 

We now apply Lemma~\ref{thm: benade et al}, setting $r = \frac{1}{2}$. Let $c, T_0$ to be the constants such that for any $T \geq T_0$, the adversary can guarantee $\ENVY(A) \geq c(T/n)^{r/2} = c(T/n)^{1/4}$. We then take any $T' \geq T_0$ where $c(T'/n)^{1/4} > x$ and set $D$ to be the uniform distribution over $\{0, 1, \nu_1, \ldots, \nu_{T'}\}$. 

We claim that for any $T \geq T'$, there is a positive probability that the allocation will not be envy-free up to $x$ goods. First, for the given algorithm, there is a positive probability that the first $T$ items drawn will follow the adversary strategy. Then, for the remaining items, there is a positive probability that 
all the items will have no value to any agent. Thus, for the final allocation $A$, we will have $\ENVY(A) > x$. Finally, observe that when item values are bounded within $[0, 1]$, $\ENVY(A) > x$ implies that $A$ is not envy-free up to $x$ goods.

\section{Missing Proofs}

\subsection{Proof of Proposition~\ref{prop:util-max-properties}}\label{app: missing proposition}

\begin{proof}[Proof of Proposition~\ref{prop:util-max-properties}]
Allocating to the agent with the highest value clearly generates a Pareto efficient allocation as it maximizes social welfare.
The algorithm first handles the special case where $D$ is a point mass. In this case $v_{it} = v$ for all $i$ and all $t$, therefore allocating the arriving items in a round-robin manner achieves an EF1 allocation, as each agent will have at most one fewer item than other agents.

In the remaining cases, we apply Theorem~\ref{thm: john et al}. The statement and proof of Theorem~\ref{thm: john et al} assume that the distribution is non-atomic, which ensures that $\arg\max_{k \in \agents} v_{kt}$ is exactly one agent and makes the utilitarian allocation algorithm well defined. In contrast, our setting assumes atomic distributions. However, the theorem can be easily adapted to our setting by adding random tie-breaking to the utilitarian allocation algorithm as described in Algorithm~\ref{alg:util}. We then rewrite Theorem~\ref{thm: john et al}, replacing [$\arg\max_{k \in \agents} v_{kt} = \{ j \}$] with [$j$ receives $t$] and similarly with [$\arg\max_{k \in \agents} v_{kt} = \{ i \}$]. This modified theorem can be shown using essentially the same proof as the original.

We now show that the distribution satisfies the two properties. The first property to show is that that for each agent $i \in \agents$, $\Pr[i \text{ receives } t] = 1/n$. This property is satisfied by symmetry, as all agent values are drawn from the same distribution and tie-breaking is done randomly. The second property is that for some constant $\Delta > 0$, and for all agents $i, j \in \agents$ where $i \neq j$: $\E[v_{it} \mid i \text{ receives } t] - \E[v_{it} \mid j \text{ receives } t] \geq \Delta$. 
We largely follow the proof of Lemma 3.2 by \citet{kurokawa2016can} with some simplifications and importantly, handling the case when the distribution $D$ is discrete.

Since agents' value distributions are identical, we can restate this as:
$$\E[v_{it} \mid i \text{ receives } t] - \E[v_{it} \mid i \text{ does not receive t}] \geq \Delta$$

Agents are all identical and the algorithm always allocates the item to an agent with the maximum value for it. This implies that $\E[v_{it} \mid i \text{ receives } t] = \E[\max(v_{1t}, \ldots, v_{nt})]$. Next, if $D$ is not a point mass, we know that $\Var[D] > 0$. From here, we can show that $\E[\max(v_{1t}, \ldots, v_{nt})] > \E[v_{it}]$.

Let $\Var[D] = c$.  Let $\bar{X} = \E[X], p = \Pr[  X < \bar{X}]$. Observe that when $\Var[D] > 0, p \in (0, 1)$.

\begin{align*}
    c &= \Var[D]\\
    &= \E[(X - \bar{X})^2]\\
    &\leq \E[\abs{X - \bar{X}}]\\
    &=p\E[\bar{X} - X \mid X < \bar{X}] + (1-p)\E[X - \bar{X} \mid X \geq \bar{X}]\\
    &=-p\E[X \mid X < \bar{X}] + (1 - p)\E[X \mid X \geq \bar{X}] + (2p - 1)\bar{X}
\end{align*}
From here, we do case-analysis based on $p \leq 1/2$ or $p > 1/2$.

Suppose that $p > 1/2$, we use the substitution $\bar{X} = p\E[X \mid X < \bar{X}] + (1 - p)\E[X \mid X \geq \bar{X}]$ and then we rewrite the above as:
\begin{align*}
    c &\leq 2p(\bar{X} - \E[X \mid X < \bar{X}])\\
    \implies \frac{c}{2} &\leq \bar{X} - \E[X \mid X < \bar{X}] \leq \E[X \mid X \geq \bar{X}] - \E[X \mid X < \bar{X}] 
\end{align*}
Similarly, we $p \leq 1/2$, we go in the other direction and get:
\begin{align*}
    c &\leq 2(1-p)(\E[X \mid X \geq \bar{X}] - \bar{X})\\
    \implies \frac{c}{2} &\leq \E[X \mid X \geq \bar{X}] - \bar{X} \leq \E[X \mid X \geq \bar{X}] - \E[X \mid X < \bar{X}] 
\end{align*}

Finally, we have that:
\begin{align*}
    \E[\max(v_{1t}, \ldots, v_{nt})] &\geq (1 - p^n)\E[X \mid X \geq \bar{X}] + p^n \E[X \mid X < \bar{X}]\\
    &\geq (1 - p)\E[X \mid X \geq \bar{X}] + p \E[X \mid X < \bar{X}]\\
    &\qquad + (p - p^n)(\E[X \mid X \geq \bar{X}] - \E[X \mid X < \bar{X}])\\
    &\geq (1 - p)\E[X \mid X \geq \bar{X}] + p \E[X \mid X < \bar{X}] + (p - p^n)\frac{c}{2}\\
    &=\E[v_{it}] + (p - p^n)\frac{c}{2}
\end{align*}

Thus, we have that $\E[v_{it} \mid i \text{ receives } t] \geq \E[v_{it}] + (p - p^n)\frac{c}{2}$. From law of total expectation, we know that $\E[v_{it}] = \frac{1}{n}\E[v_{it} \mid i \text{ receives } t] + \frac{n-1}{n}\E[v_{it} \mid i \text{ does not receive } t]$. We can combine and rearrange to also show that $\E[v_{it} \mid i \text{ does not receive } t] \leq \E[v_{it}] - (p - p^n)\frac{c}{2(n-1)}$, which allows us to conclude:

$$\E[v_{it} \mid i \text{ receives } t] - \E[v_{it} \mid i \text{ does not receive } t] \geq (p - p^n)\frac{c}{2(n-1)} + (p - p^n)\frac{c}{2}$$

Setting $\Delta$ to $(p - p^n)\frac{c}{2(n-1)} + (p - p^n)\frac{c}{2}$, which is positive since $p \in (0, 1)$ and $c > 0$, completes the proof.\end{proof}

%\subsection{Missing from Theorem~\ref{thm:por-po}}\label{app:missing from po rounding}

%Define 
%    \[
%    c^* = \min \{ \min_{i \in \agents, j \in [m]: Y_{ij} > Y'_{ij}} \frac{X_{ij}}{Y_{ij} - Y'_{ij}} , \min_{i \in \agents, j \in [m]: Y_{ij} < Y'_{ij}} \frac{1-X_{ij}}{Y'_{ij} - Y_{ij}} \} .
%    \]
%    $Y$ and $Y'$ are both feasible and strictly different allocations, so there must be some agent $i$ and item type $\gamma_j$ such that $Y_{ij} > Y'_{ij}$ and other some other pair $\hat{i},\hat{j}$ such that $Y_{\hat{i}\hat{j}} > Y'_{\hat{i}\hat{j}}$. Furthermore, since $Y_{ij} > 0$ implies that $X_{ij} > 0$ (recall that an item $t$ of type $\gamma_j$ is in $A_i$ only if $X_{ij} > 0$), we have that $c^*>0$. It remains to show that $X + c^*\Delta$ is feasible. There are three cases for an agent $i$ and item type $\gamma_j$: (1) $Y_{ij} = Y'_{ij}$, (2) $Y_{ij} > Y'_{ij}$, (3) $Y_{ij} < Y'_{ij}$. For case (1), trivially $(X + c^*\Delta)_{ij} = X_{ij} \in [0,1]$. For case (2), $(X + c^*\Delta)_{ij} < X_{ij} \leq 1 $, and $(X + c^*\Delta)_{ij} \geq X_{ij} + \frac{X_{ij}}{Y_{ij} - Y'_{ij}} \cdot (Y'_{ij} - Y_{ij} ) = 0$. Finally, for case (3), $(X + c^*\Delta)_{ij} > X_{ij} \geq 0 $, and $(X + c^*\Delta)_{ij} > X_{ij} + \frac{1-X_{ij}}{Y'_{ij} - Y_{ij}} \cdot (Y'_{ij} - Y_{ij}) \leq 1$.

\subsection{Proof of Lemma~\ref{lem:edge-indifferent}}\label{app: missing lemma proof}

We can rewrite KKT condition~\eqref{eq: KKT 3} as: for all $i \in [n]$, for all $k \in X_i$, $p_j = \frac{v_{ik} \cdot e_i}{\sum_{k' \in \xalloc_i} v_{ik'}x_{ik'}} = \frac{v_{ik}}{r_i}$.

Then, for every agent $i$ we have $\sum_{k \in \xalloc_i} p_k x_{ik} = \sum_{k \in \xalloc_i} \frac{v_{ik} \cdot e_i}{\sum_{k' \in \xalloc_i} v_{ik'}x_{ik'}}x_{ik} = e_i$

For the ``only if'' direction, we know for all items $k \in \xalloc_j$, $v_{ik} = r_i p_k$. We can substitute this into our previous equation to get:
$$v_i(\xalloc_i) = \sum_{k \in \xalloc_i} v_{ik} x_{ik} = \sum_{k \in \xalloc_i} r_i p_k x_{ik} = r_i e_i = r_i e_j = \sum_{k \in \xalloc_j} r_i p_k x_{jk} = \sum_{k \in \xalloc_j} v_{ik} x_{jk} = v_i(\xalloc_j).$$

For the ``if'' direction, assume that there is an item $k^* \in \xalloc_j$ such that $\frac{v_{i k^*}}{p_{k^*}} \neq r_i$. KKT condition~\eqref{eq: KKT 3} implies that $k^*$'s bang-per-buck is at most $r_i$, so we must have $v_{ik^*} < r_ip_{k^*}$. Combined, we get:
\[
v_i(X_i) = \sum_{k \in X_i} v_{ik} x_{ik} = \sum_{k \in X_i} r_i p_k x_{ik} = r_i e_i = r_i e_j = \sum_{k \in X_j} r_i p_k x_{jk} > \sum_{k \in X_j} v_{ik} x_{jk} = v_i(X_j). \qed
\]

\subsection{Proof of Lemma~\ref{lem:optimal-transfers}}\label{app: lemma optimal transfers}

    $(\bx, \bp, \be)$ must satisfy the KKT conditions. $\bp$ does not change and $\bx'$ is feasible by definition; it remains to show that $(\bx', \bp, \be')$ satisfies the KKT conditions. 
    Condition~\eqref{eq: KKT 1} is satisfied since $\sum_{i \in [n]} \Delta_{ik} = 0$ for all items $k$. 
Furthermore, notice that  $\Delta_{ik} \neq 0$ implies that $\frac{v_{ik}}{p_k} = r_i$: if $\Delta_{ik} > 0$ this fact is implied by the definition of an optimal transfer, while if $\Delta_{ik} < 0$, for $\bx + \Delta$ to be feasible, it must be that $x_{ik} > 0$, so $\frac{v_{ik}}{p_k} = r_i$ is implied by condition~\eqref{eq: KKT 3}.
Thus:   
\[ \frac{\sum_{k=1}^m v_{ik}x'_{ik}}{e'_i} = \frac{\sum_{k=1}^m v_{ik}x_{ik} + \sum_{k=1}^m v_{ik}\Delta_{ik}}{e_i + \delta_i} = \frac{\sum_{k=1}^m v_{ik}x_{ik} + \sum_{k=1}^m r_i p_k \Delta_{ik}}{e_i + \delta_i}  = \frac{\sum_{k=1}^m v_{ik}x_{ik} + r_i \delta_i}{e_i + \delta_i},  \] 
which is equal to $\frac{\sum_{k=1}^m v_{ik}x_{ik}}{e_i},$ since $r_i = v_i(\xalloc_i)/e_i$. Therefore, the RHS of condition~\eqref{eq: KKT 2} does not change (the LHS of course didn't change since it only has values and prices), so condition~\eqref{eq: KKT 2} is still satisfied. When $\bx_{ik} > 0$, condition~\eqref{eq: KKT 3} is satisfied by similar reasoning. Finally, it is possible that $\bx_{ik} = 0$ but $\bx'_{ik} > 0$. In this case, we know that $\Delta_{ik} > 0$, and therefore $\frac{v_{ik}}{p_k} = r_i = \frac{\sum_{k=1}^m v_{ik}x_{ik}}{e_i},$ by the definition of an optimal transfer and the definition of $r_j$. Thus, condition~\eqref{eq: KKT 3} is satisfied.
\qed

\subsection{Choosing $b$}\label{app:choosing b}

\paragraph{Operation $1$.}
We first discuss how to choose $b$ in Operation $1$ from Lemma~\ref{lem:const-util-shifts}.

For any pair $(i, j)$, we are concerned about $v_i(A_i) - v_i(A_j)$. Operation $1$ guarantees that $v_i(X_i) = v_i(X'_i)$, so $v_i(A_i) - v_i(A_j)$ can change only when $j \in K$. Define $c = \max_{i \in \agents,k \in [m]} \frac{v_{ik}}{p_k}$: the maximum bang-per-buck for any agent. Then, we choose $b$ such that:

$$b < \min \left\{ \frac{v_i(A_i) - v_i(A_j)}{c} : i \in \agents, j \in K, v_i(A_i) - v_i(A_j) > 0 \right\}$$

Observe that a budget constraint of $b$ ensures that $v_i(A_j)$ will change by at most $b \cdot c < v_i(A_i) - v_i(A_j)$. 

\paragraph{Lemma~\ref{lem:budget-change-shifts}.}
Next, we discuss how to choose $b$ in the procedure described in Lemma~\ref{lem:budget-change-shifts}.

We use a similar approach. Unlike before, $v_i(A_i) - v_i(A_j)$ can change when either $i \in S$ or $j \in S$, as $v_i(A_i)$ changes for some $i \in S$. Setting

$$b < \min \left\{ \frac{v_i(A_i) - v_i(A_j)}{2c} : i, j \in \agents, i \in S \vee j \in S, v_i(A_i) - v_i(A_j) > 0 \right\}$$

as the budget constraint is sufficient as it ensures that $v_i(A_i)$ and $v_i(A_j)$ can each change by at most $\frac{v_i(A_i) - v_i(A_j)}{2}$.

\subsection{Proof of Theorem~\ref{thm:eff-envy-lb-det}}\label{app: lower bound proof continued}
\begin{proof}[Proof of Theorem~\ref{thm:eff-envy-lb-det} (continued) ]
Suppose that there is an allocation algorithm which guarantees that for all $T$, no matter the instance the adversary selects, $\E[ \ENVY(\alloc^T) ] \leq f(T)$ for some $f(T) \in o(T)$, where the expectation is over the randomness used by the algorithm. We first describe a family of $n$ instances. For $i=1$ to $n$, instance $I_i$'s first $\frac{T}{n}i$ items follow $I$, the instance of the adaptive adversary described above, and the remaining items have no value. Let $g(T)$ be a function such that $g(T)\cdot f(T) \in o(T)$ and $g(T) \in \omega(1)$\footnote{One can think of $g(T) = T^\delta$, for some small $\delta > 0$ that depends on $f(T)$.}.

Again, we bound the number of items the algorithm can allocate to each agent in each segment; this time our bounds will be looser and probabalistic. Consider the behavior of the algorithm when faced with instance $I_1$. At the end of the first segment, i.e. for items $1$ through $\frac{T}{n}$, if the algorithm allocates to some agent $k$ at least $\frac{T}{n^2} + x_1$, for $x_1>0$, with probability at least $\frac{1}{g(T)}$, then the expected envy of some agent $\hat{k}$ (the one who received fewer than $\frac{T}{n^2}$ items) is at least $\varepsilon \cdot x_1$, that is $\E[ \ENVY(\alloc^T) ] \geq \frac{1}{g(T)} \cdot \varepsilon x_1$. Since $\E[ \ENVY(\alloc^T) ] \leq f(T)$, we have that $x_1 \leq \frac{g(T) f(T)}{\varepsilon}$. In other words, with probability $1 - \frac{1}{g(T)}$, each agent receives a number of items within $[\frac{T}{n^2} - x_1, \frac{T}{n^2} + x_1]$. Note that because the first $\frac{T}{n}$ items are identical for all instances, this holds for instances $I_2, \ldots, I_n$. 

Similarly, we consider the behavior of the algorithm faced with $I_2$ and look at the end of the second segment. Suppose that conditioned on the algorithm having allocated each agent a number of items within $[\frac{T}{n^2} - x_1, \frac{T}{n^2} + x_1]$ after the first segment, with (conditional) probability at least $\frac{1}{g(T) - 1}$, some agent $k$ receives at least $\frac{T}{n^2} + x_2$ items from the second segment. This translates to an unconditioned probability of occuring of at least $\frac{1}{g(T)}$ and for similar reasons as before, we must have that $x_2 \leq \frac{g(T)f(T)}{\varepsilon} + \frac{2x_1}{\varepsilon} \leq \frac{g(T) f(T)}{\varepsilon} \left( 1 + \frac{2}{\varepsilon} \right)$. Together, we know that with probability $1 - \left(\frac{1}{g(T)} + \frac{1}{g(T) - 1}\left(1 - \frac{1}{g(T)}\right)\right)= 1 - \frac{2}{g(T)}$, each agent receives a number of items within $[\frac{T}{n^2} - x_2, \frac{T}{n^2} + x_2]$. We continue this for larger $i$, for $x_i = \frac{g(T) f(T)}{\varepsilon} \left( 1 + \frac{2}{\varepsilon} \right)^{i-1}$. Finally, we analyze efficiency of the algorithm for instance $I_n$. Each agent receives utility at most $(1 + (n - 1)\varepsilon))(\frac{T}{n^2} + x_n) + \frac{n}{g(T)}T$, where the additional $\frac{n}{g(T)}T$ term accounts for the worst case allocation assuming a deviation. 

$$\frac{u_i}{1/n + \varepsilon} < \frac{T}{n} \cdot \frac{1 + (n - 1)\varepsilon}{1 + \varepsilon n} \cdot \left(1 + \frac{n^2}{\varepsilon} \frac{f(T)g(T)}{T}(1 + \frac{2}{\varepsilon})^{n - 1} + n^3g(T)\right).$$

By picking $T$ large enough\footnote{Ensuring $\frac{g(T)f(T)}{T} < \frac{1}{2}\frac{\epsilon}{1+(n-1)\varepsilon} \cdot \frac{\varepsilon}{n^2 (1 + 2/ \varepsilon)^{n - 1} }$ and $g(T) < \frac{\varepsilon}{2n^3}$ will suffice.} the adversary can make sure that the expected utility for every agent $i$ is upper bounded by $\frac{T}{n} \cdot ( 1/n + \varepsilon )$ . On the other hand, the single allocation that gives items from $\frac{T}{n}(i - 1)$ to $\frac{T}{n}i$ to agent $i$ yields utility $u_i = \frac{T}{n}$.
\end{proof}

\section{Strong Envy-Free and Pareto Optimality are not Achievable}
\label{app:sef-po-nonident-cex}
Here, we give an instance where strong envy-freeness and Pareto optimality are not achievable even though agents valuations are not identical up to a multiplicative factor. 
\renewcommand{\arraystretch}{1}

\setlength\tabcolsep{2pt}
\begin{table*}[ht]
\begin{center}
    \[\begin{array}{ | c | c c c | }
    \hline
     \text{Item $t$} & g_1 & g_2 & g_3 \\ [0.5ex] 
     \hline
     \text{Value of $g_t$ to Agent 1} & 1 & 1 & 1 \\
     \text{Value of $g_t$ to Agent 2} & 0.5 & 1 & 1\\
     \text{Value of $g_t$ to Agent 3} & 0.25 & 1 & 1 \\
     \hline
    \end{array}\]
    \caption{Instance with three agents}
    \label{table:efpareto-incompatible}
\end{center}
\end{table*}

In Table~\ref{table:efpareto-incompatible}, none of the three agents are identical. However, we claim that in any envy-free and Pareto optimal allocation, agents $2$ and $3$ will be indifferent towards each other's allocation. Intuitively, the problem is that agents $2$ and $3$ have identical valuations over the items they could possibly receive in an envy-free and Pareto efficient allocation, items $2$ and $3$. 

Let $X$ be any Pareto efficient and envy-free allocation. Since it is envy-free and agent $1$ is indifferent between the items, we must $X_{11} + X_{12} + X_{13} \geq 1$. 

Next, we show that if $X$ is Pareto efficient, then $X_{11} = 1$. If $X_{11} < 1$, then either $X_{12} > 0$ or $X_{13} > 0$. In addition, either $X_{21} > 0$ or $X_{31} > 0$. Without loss of generality, suppose $X_{12} > 0$ and $X_{21} > 0$. Then letting $c = \min(X_{21}, X_{12})$ we give the following allocation $X'$ which Pareto dominates $X$ so $X$ is not Pareto efficient.

\begin{equation*}
    X' = X + \begin{bmatrix} 
        c & -c & 0\\
        -c & c & 0\\
        0 & 0 & 0 
        \end{bmatrix}
\end{equation*}

Therefore, we know $X_{11} = 1$ and as a result, $X_{21} = X_{31} = 0$. Finally, because agents $2$ and $3$ have identical valuations for items $2$ and $3$, we have $X_{22} + X_{23} \geq 1$ and $X_{32} + X_{33} \geq 1$. Together, this implies that $X_{22} + X_{23} = X_{32} + X_{33} = 1$ so agents $2$ and $3$ will be indifferent towards each other's allocations.

\section{Deriving KKT Conditions}
\label{app:kkt-derive}
Here, we derive the KKT conditions in order to show that they are both necessary and sufficient conditions for optimal solutions to the Eisenberg-Gale convex program.

We first introduce dual variables $\bp, \bk$ for the first and second inequality constraints. From stationarity, we have:
$$\grad_{\bx}\sum_{i = 1}^n e_i \log \sum_{j = 1}^m v_{ij}x_{ij} = \sum_{j = 1}^m p_j \grad_{\bx} \sum_{i = 1}^n x_{ij} - 1 + \sum_{i = 1}^n\sum_{j = 1}^m - k_{ij} \grad_{\bx} x_{ij}$$

For each $i, j$, we can take the gradient with respect to $x_{ij}$ to get:
\begin{align}
    e_i\frac{v_{ij}}{\sum_{j = 1}^m v_{ij}x_{ij}} &= p_j - k_{ij} \nonumber \\
    \frac{v_{ij}}{p_j} &= \frac{\sum_{j = 1}^m v_{ij}x_{ij}}{e_i}\left(1 - \frac{k_{ij}}{p_j}\right) \label{ln:stationarity}
\end{align}

The primal and dual feasibility conditions tell us that:
$$-x_{ij} \leq 0, \sum_{i=1}^n x_{ij} - 1 \leq 0, p_j \geq 0, k_{ij} \geq 0$$

Finally, complementary slackness tells us that:
\begin{align}
    x_{ij} > 0 \implies k_{ij} = 0 &\text{ and } k_{ij} > 0 \implies x_{ij} = 0 \label{ln:comp-slack}\\
    p_j > 0 \implies \sum_{i=1}^n x_{ij} = 1 &\text{ and } \sum_{i=1}^n x_{ij} < 1 \implies p_j = 0 \label{ln:comp-slack-2}
\end{align}

KKT condition~\ref{eq: KKT 1} follows from Equation~\ref{ln:comp-slack-2}.
Meanwhile, we show that KKT conditions~\ref{eq: KKT 2} and~\ref{eq: KKT 3} are equivalent to the stationarity condition plus the first two complementary slackness conditions. 

\begin{proposition}
    For any $\bx, \bp$ where $p_j > 0$, KKT conditions~\ref{eq: KKT 2} and~\ref{eq: KKT 3} hold if and only if there exists $\bk$ such that Equations~\ref{ln:stationarity} and~\ref{ln:comp-slack} hold. 
\end{proposition}
\begin{proof}
    Consider any $\bx, \bp$.

    We first show the forwards direction. We assume that KKT conditions~\ref{eq: KKT 2} and~\ref{eq: KKT 3} hold and give a $\bk$ such that Equations~\ref{ln:stationarity} and~\ref{ln:comp-slack} hold. Take any $i \in [n], j \in [m]$. We set the value of $k_{ij}$ depending on whether $x_{ij} > 0$. Suppose that $x_{ij} > 0$. Then observe that setting $k_{ij} = 0$ will satisfy both the stationarity and complementary slackness conditions. Otherwise, $x_{ij} = 0$. In this case, the slackness conditions trivially hold and we just need to show that there exists a $k_{ij} \geq 0$ such that Equation~\ref{ln:stationarity} holds. From KKT condition~\ref{eq: KKT 2}, we have:
    $$\frac{v_{ij}}{p_j} / \frac{\sum_{j = 1}^m v_{ij}x_{ij}}{e_i} \leq 1$$

    Letting $c = \frac{v_{ij}}{p_j} / \frac{\sum_{j = 1}^m v_{ij}x_{ij}}{e_i}$, we solve for $k_{ij}$ which gives $k_{ij} = p_j (1 - c)$, which is non-negative.

    Next, we show the reverse direction. Assume that there exists $\bk$ such that Equations~\ref{ln:stationarity} and~\ref{ln:comp-slack} hold. For any $i \in [n], j \in [m]$, because we have that $p_j > 0$ and $k_{ij} \geq 0$, this implies that $1 - \frac{k_{ij}}{p_j} \leq 1$. Combine this with Equation~\ref{ln:stationarity} and we get exactly KKT condition~\ref{eq: KKT 2}. 

    Next, we show that KKT condition~\ref{eq: KKT 3} holds. Assume that $x_{ij} > 0$. Equation~\ref{ln:comp-slack} tells us that $k_{ij} = 0$. Therefore, $1 - \frac{k_{ij}}{p_j} = 1$ and applying Equation~\ref{ln:stationarity} tells us that $\frac{v_{ij}}{p_j} = \frac{\sum_{j = 1}^m v_{ij}x_{ij}}{e_i}$. 
\end{proof}

\section{Improving Guarantee for Independent Agents}
\label{app:independent agents strong ef}

When the adversary is limited to distributions where agents' valuations are independent, under the mild assumption that the distributions $D_i$ are not point masses, we give an ex-post Pareto efficient allocation algorithm that guarantees envy-freeness with high probability. That is, it is possible to improve the ``or EF1'' part of the guarantee of Theorem~\ref{thm:POCR} when agents are independent and the distributions are not point masses.

Our approach builds on our structural result  (Theorem~\ref{thm: structural}) obtained for correlated agents, which transforms the online problem into one of finding a fractional allocation for divisible items. %For our main result for correlated agents, we simply reduced the independent agents case to the correlated agents case, as we can always multiply the independent distributions to form a joint distribution, or alternatively, a distribution over items. However, 
For independent agents, we can find a fractional allocation that is strongly envy-free and Pareto optimal. Under this stronger guarantee, the proof of Theorem~\ref{thm:POCR} gives that when the graph is strongly envy-free, using Algorithm~\ref{alg:POCR} guarantees envy-freeness with high probability.

\subsection{Strong envy-freeness example}
At first, it seems plausible that with independent agents and distributions that are not point masses, any solution to the E-G convex program will be strongly envy-free and Pareto optimal. Unfortunately, this is not the case. In this section, we give an example where there exists a fractional allocation that is both Pareto optimal and strongly envy-free. However, we show that this solution can be missed when choosing an arbitrary solution to the E-G convex program, or even by using the algorithm for finding CISEF and Pareto optimal solutions. We show how the E-G solution can be transformed into a strongly envy-free and Pareto optimal allocation, highlighting the key ideas behind our algorithm.

Consider the following fractional allocation problem for three agents. The agents' value distributions are as follows: For both $X \sim D_1$ and $X \sim D_2$, we have $\Pr[X = 0] = \frac{1}{10}, \Pr[X = 1] = \frac{9}{10}$. For $X \sim D_3$, we have $\Pr[X = 1] = \frac{16}{17}, \Pr[X = 2] = \frac{1}{17}$. 

Recall that in the fractional allocation problem, we treat each item type $\gamma_j$ as a divisible item. Each agent $i$ has a value function $v'_i$, where $v'_i(\gamma_j) = v_i(\gamma_j) \cdot f_D(\gamma_j)$. We now present a solution $\bx, \bp$ to the E-G convex program with identical budgets of $1$, shown in Table~\ref{table:independent-agents-nonstrong}.

\setlength\tabcolsep{2pt}
\begin{table*}[ht]
\begin{center}
    \[\begin{array}{ | c | c | c | c | c | c | c | c | c | c | c | c | }
    \hline
     \text{Item type $\gamma_j$} & v_{1t} & v_{2t} & v_{3t} & f_D(\gamma_j) & v'_{1t} & v'_{2t} & v'_{3t} & x_{1t} & x_{2t} & x_{3t} & p_t\\ [0.5ex] 
     \hline
      \gamma_1 & 0 & 0 & 1 & \frac{16}{1700} & 0 & 0 & \frac{16}{1700} & 0 & 0 & 1 & \frac{16}{600}\\
      \gamma_2 & 0 & 0 & 2 & \frac{1}{1700} & 0 & 0 & \frac{2}{1700} & 0 & 0 & 1 & \frac{2}{600} \\
      \gamma_3 & 0 & 1 & 1 & \frac{144}{1700} & 0 & \frac{144}{1700} & \frac{144}{1700} & 0 & 0 & 1 & \frac{144}{600}\\
      \gamma_4 & 0 & 1 & 2 & \frac{9}{1700} & 0 & \frac{9}{1700} & \frac{18}{1700} & 0 & 0 & 1 & \frac{18}{600} \\
      \gamma_5 & 1 & 0 & 1 & \frac{144}{1700} & \frac{144}{1700} & 0 & \frac{144}{1700} & 0 & 0 & 1 & \frac{144}{600}\\
      \gamma_6 & 1 & 0 & 2 & \frac{9}{1700} & \frac{9}{1700} & 0 & \frac{18}{1700} & 0 & 0 & 1 & \frac{18}{600} \\
      \gamma_7 & 1 & 1 & 1 & \frac{1296}{1700} & \frac{1296}{1700} &  \frac{1296}{1700} &  \frac{1296}{1700} & \frac{75}{162} & \frac{75}{162} & \frac{12}{162} & \frac{1296}{600}\\
      \gamma_8 & 1 & 1 & 2 & \frac{81}{1700} & \frac{81}{1700} & \frac{81}{1700} & \frac{162}{1700} & 0 & 0 & 1 & \frac{162}{600} \\
    \hline
    \end{array}\]
    \caption{Fractional allocation problem. All items with positive value to an agent are also MBB items for that agent. The MBB ratio for agents $1$ and $2$ is $\frac{3}{17}$ and for agent $3$ is $\frac{6}{17}$.}
    \label{table:independent-agents-nonstrong}
\end{center}
\end{table*}

$\bx, \bp$ are constructed such that each agent spends their entire budget of $1$ on their items. $\bp$ has the property that all items are bang-per-buck items for agent $3$. Agent $3$ has maximum bang-per-buck items $\gamma_5, \gamma_7$ while agent $2$ has MBB items $\gamma_3, \gamma_7$. The allocation $\bx$ satisfies the maximum bang-per-buck constraints. Therefore, $\bx, \bp$ is a solution.

The envy-graph for $\bx$ still contains four edges: $(1, 2), (2, 1), (3, 1), (3, 2)$. Now consider what happens if we apply the algorithm for finding a CISEF and Pareto optimal solution. Procedure $2$, as described in Lemma~\ref{lem:budget-change-shifts}, increases the budget of agent $3$ and gives agent $3$ more of $\gamma_7$. Concretely, we might increase the budget of agent $3$ by $\frac{16}{600}$ and decrease the budget of agents $1$ and $2$ by $\frac{8}{600}$, resulting in a new solution $(\bx' = \xalloc', \bp')$ where $\bp' = \bp$ and the only differences between $\bx'$ and $\bx$ are that $\bx'_{17} = \bx'_{27} = \frac{74}{162}$ and $\bx'_{37} = \frac{14}{162}$. The resulting envy-graph is now CISEF but not strongly envy-free, as it contains two edges: $(1, 2), (2, 1)$.

In this example, it is fairly easy to see how to remove the remaining two envy edges and maintain Pareto optimality. Suppose we transfer $\Delta$ worth of $\gamma_7$ from agent $1$ to $3$ and $\Delta$ worth of $\gamma_5$ from agent $3$ to $1$. We do a similar transfer for agents $2$ and $3$ with items $\gamma_7$ and $\gamma_3$. As agent $1$ does not value $\gamma_3$ and agent $2$ does not value $\gamma_5$, performing these two changes will remove the two envy edges. If we choose $\Delta$ small enough, no new envy edges are created.

\subsection{Strong Envy-freeness and Pareto Optimality}

The above approach can be generalized into an algorithm for finding a strongly envy-free and Pareto optimal allocation. For each distribution $D_i$, let $S_i$ be the support of $D_i$. The independent agent adversary can be simulated by the correlated agent adversary by expanding out the individual distributions into a distribution for item types $D$ with support $G_D$ where $G_D = S_1 \times S_2 \times \cdots \times S_n$. In this section, we will use two notations, $\gamma_j$ and $(a_1, \ldots, a_n)$, interchangeably to refer to an item type. As described in the Pareto optimal clique rounding algorithm, we define a new valuation function $v'_i$ for each agent $i$: $v'_i((a_1, \ldots, a_n)) = v_i((a_1, \ldots, a_n))f_D((a_1, \ldots, a_n)) = a_i \cdot f_D((a_1, \ldots, a_n))$. 

\begin{theorem}
\label{thm:ind-sef}
Given an instance with items $G_D = S_1 \times S_2 \times \cdots \times S_n$, where for all $i$, $\abs{S_i} \geq 2$, and $n$ agents with valuation functions $v'_i((a_1, \ldots, a_n)) = a_i \cdot f_D((a_1, \ldots, a_n))$, there always exists an allocation that is strongly envy-free and Pareto efficient.
\end{theorem}

Applying Theorem~\ref{thm: structural}, we start with a solution $(\bx = \xalloc, \bp)$ to budgets $\be$ where $\xalloc$ is CISEF and Pareto optimal. Now suppose in $I(X)$ there exists a clique $K$ with size $k$. Without loss of generality, suppose that these are the first $k$ agents. For agents $i \in K$, let $h_i = \max S_i$ and $l_i = \min S_i$. These are agent $i$'s highest and lowest values for items prior to scaling by the probability of the item. Since $\abs{S_i} \geq 2$, $h_i > l_i$.

For any $(a_1, a_2, \ldots, a_k) \in S_1 \times S_2 \ldots \times S_k$, we denote the subset of items with those values with $G_{(a_1, a_2, \ldots, a_k)} = \{(b_1, b_2, \ldots, b_n) : b_1 = a_1, b_2 = a_2, \ldots, b_k = a_k, (b_1, b_2, \ldots, b_n) \in G_D \}$. In particular, we care about the subsets of items involving the high and low values for each agent. Let $H = (h_1, \ldots, h_k)$ and for each $i \in K$, let $H^i = (l_1, \ldots, l_{i - 1}, h_i, l_{i + 1}, \ldots l_k)$. The corresponding subsets with those values are $G_H$ and $G_{H^i}$.

For a $\gamma_j = (a_1, \ldots, a_n)$, we denote to the item's value to the remaining agents as $R(\gamma_j) = (a_{k+1}, \ldots, a_n)$. In particular, we are interested in the items in $G_H$ and $G_{H^i}$ with values $R(\gamma_j)$. Define $(H^i, R(\gamma_j)) = (l_1, \ldots, l_{i - 1}, h_i, l_{i+1}, \ldots l_k, a_{k + 1}, \ldots, a_n)$ and $(H, R(\gamma_j)) = (h_1, \ldots, h_k, a_{k+1}, \ldots, a_n)$.

Now consider the allocation of the agents in $K$. From the CISEF property, they all have the same allocation, denoted as $X_K$. We first show the following two claims:
 
\begin{lemma}
\label{lem:high-value-item}
There exists an item $\gamma_j \in G_H \cap X_K$.
\end{lemma}

\begin{lemma}
\label{lem:high-low-item}
For all items $\gamma_j \in X_K$, $\gamma_j \notin G_{H^1} \cup G_{H^2} \cup \cdots \cup G_{H^k}$.
\end{lemma}

\begin{proof}[Proof of Lemma~\ref{lem:high-value-item}]
We first show that for any $\gamma_j$ and Pareto optimal $X$, if $\gamma_j \in X_K$, then $(H, R(\gamma_j)) \in X_K$. Assume for sake of contradiction that $\gamma_j \in X_K$ but $(H, R(\gamma_j)) \notin X_K$. Let $\gamma_{j'} = (H, R(\gamma_j))$ and let agent $i$ be any agent not in $K$ where $\gamma_{j'} \in X_i$. From KKT condition~\eqref{eq: KKT 3}, item $j'$ must maximize agent $i$'s bang-per-buck, meaning:

$$\frac{v_{ij} \cdot f_D(\gamma_{j})}{p_{j}} = \frac{v'_{ij}}{p_{j}} \leq \frac{v'_{ij'}}{p_{j'}} = \frac{v_{ij'} \cdot f_D(\gamma_{j'})}{p_{j'}}$$

However, because $v_{ij} = v_{ij'}$, this implies that $\frac{f_D(\gamma_{j})}{p_{j}} \leq \frac{f_D(\gamma_{j'})}{p_{j'}}$. Meanwhile, item $j$ maximizes agent $1$'s bang-per-buck so:

$$\frac{h_1 \cdot f_D(\gamma_{j'})}{p_{j'}} = \frac{v_{1j'} \cdot f_D(\gamma_{j'})}{p_{j'}} = \frac{v'_{1j'}}{p_{j'}} \leq \frac{v'_{1j}}{p_{j}} = \frac{v_{1j} \cdot f_D(\gamma_{j})}{p_{j}} < \frac{h_1 \cdot f_D(\gamma_{j})}{p_{j}}$$

This implies that $\frac{f_D(\gamma_{j'})}{p_{j'}} < \frac{f_D(\gamma_j)}{p_j}$, contradicting the previous inequality. We know $X_K$ is non-empty, as otherwise the allocation $X$ would not be envy-free. Thus, there exists some $\gamma_j \in X_K$. This implies that $(H, R(\gamma_j)) \in X_K$ and also by definition, $(H, R(\gamma_j)) \in G_{H}$.
\end{proof}

\begin{proof}[Proof of Lemma~\ref{lem:high-low-item}]
From Lemma~\ref{lem:high-value-item}, there exists a $\gamma_j \in G_H \cap \xalloc_K$. Now consider an arbitrary $\gamma_{j'} \in G_{H^i}$ for some $i \in K$. As shown earlier in the proof of Lemma~\ref{lem:high-value-item}, if $v_{ij'} = v_{ij}$ and $\gamma_j$ is an MBB item for agent $i$, then $\frac{f_D(\gamma_{j'})}{p_{j'}} \leq \frac{f_D(\gamma_j)}{p_j}$.

However, for all other agents $i' \in K, i' \neq i$:

$$\frac{v'_{i'j'}}{p_{j'}} = \frac{v_{i'j'} \cdot f_D(\gamma_j')}{p_{j'}} \leq \frac{v_{i'j'} \cdot f_D(\gamma_j)}{p_{j}} < \frac{h_{i'} \cdot f_D(\gamma_j)}{p_{j}} = \frac{v'_{i'j}}{p_{j}}$$

In other words, item $j'$ is not a maximum bang-per-buck item for $i'$. As agents in $K$ have the same allocation $\xalloc_K$, this means that $\gamma_j \notin \xalloc_K$.

\end{proof}

Now, we introduce a third procedure to the indifference edge elimination algorithm:

\paragraph{Procedure $3$.}
Take a clique in the indifference graph and choose an item $\gamma_j \in G_H \cap X_K$, which exists by Lemma~\ref{lem:high-value-item}. We construct $X'$ by performing $k$ transfers. For each $i \in K$, by Lemma~\ref{lem:high-low-item}, there must exist some agent $i' \notin K$ where $(H^i, R(\gamma_j)) \in X_{i'}$. We define an item transfer $\Delta^i$ as: transfer $\delta^i$ worth of $(H^i, R(\gamma_j))$ from agent $i'$ to agent $i$, and $\delta^i$ worth of $\gamma_j$ from $i$ to $i'$.

This transfer has the property that both $\gamma_j$ and $(H^i, R(\gamma_j))$ are MBB items for agent $i$ and $i'$, but $(H^i, R(\gamma_j))$ is not an MBB item for all other agents $i^* \in K, i^* \neq i$. To show these properties, letting $\gamma_{j'} = (H^i, R(\gamma_j))$, we already know that $\gamma_j$ is an MBB item for agent $i$ and $\gamma_{j'}$ is an MBB item for agent $i'$. Because $v_{ij} = v_{ij'}$ and $v_{i'j} = v_{i'j'}$, we can show that $\frac{f_D(\gamma_{j'})}{p_{j'}} \leq \frac{f_D(\gamma_j)}{p_j}$ and also $\frac{f_D(\gamma_{j'})}{p_{j'}} \geq \frac{f_D(\gamma_j)}{p_j}$, implying that $\frac{f_D(\gamma_{j'})}{p_{j'}} = \frac{f_D(\gamma_j)}{p_j}$. Using this equality, we show:

$$\frac{v'_{ij}}{p_j} = \frac{v_{ij} \cdot f_D(\gamma_j)}{p_j} = \frac{v_{ij'} \cdot f_D(\gamma_{j'})}{p_{j'}} = \frac{v'_{ij'}}{p_{j'}}$$

Thus, items $j$ and $j'$ both maximize bang-per-buck for agent $i$. We can show the same for agent $i'$. Meanwhile, for agent $i^*$, the same equality shows $\gamma_{j'} = (H^i, R(\gamma_j))$ is not an MBB item:

$$\frac{v'_{i^*j'}}{p_{j'}} = \frac{v_{i^*j'} \cdot f_D(\gamma_{j'})}{p_{j'}} < \frac{v_{i^*j} \cdot f_D(\gamma_{j})}{p_{j}} = \frac{v'_{i^*j}}{p_{j}}$$

The only remaining condition for $\Delta = \sum_i \Delta^i$ to be an optimal transfer is to set the $\delta^i$ small enough that $X + \Delta$ is feasible.

Next, for any two agents $i, j \in K$, we want to show that this transfer eliminates the edge $(i, j)$. $X' = X + \sum_i \Delta^i$. Transfer $\Delta^i$ will eliminate the edge $(j, i)$ for all $j \in K, j \neq i$ since agent $i$ now owns an item that is not an MBB item for agent $j$. Finally, to ensure no new indifference edges or envy-freeness violations are created, we set $\delta^i = \frac{b}{k}$, where $b$ is chosen using the method discussed in Appendix~\ref{app:choosing b}. The extra $\frac{1}{k}$ factor ensures that we can perform all $k$ transfers safely.

One issue is that $X'$ might not be CISEF. Consider when the recipient in $\Delta^i$, agent $i'$, is part of clique $K'$ with $\abs{K'} > 1$. In this case, we change $X_{i'}$ without changing the allocation for other agents in $K'$. To handle this case, let $\xalloc_{K'}$ be the allocation for agents in $K'$. Any transfer of items to and from the agent can instead be spread out across all agents in the clique. Specifically, the final allocation is $X^*$, where $\xalloc^*_{K'} = \xalloc_{K'} + \frac{1}{\abs{K'}}(\xalloc'_{i'} - \xalloc_{i'})$ for the relevant $i', K'$ and the same as $X'$ otherwise. This allocation can be shown to be CISEF and Pareto optimal.

\begin{proof}[Proof of Theorem~\ref{thm:ind-sef}]
From Theorem~\ref{thm: structural}, we start with a solution $(\bx = \xalloc, \bp)$ to budgets $\be$ where $\xalloc$ is CISEF and Pareto optimal. Starting from allocation $\xalloc$, we repeat procedure $3$ while the indifference graph still contains cliques (of size greater than $1$). Each time we perform procedure $3$, we eliminate one clique without creating new indifference edges (or envy violations). Since procedure $3$ describes an optimal transfer, the resulting allocation is still Pareto optimal. Once we eliminate all cliques of size greater than $1$, the resulting indifference graph only contains cliques of size $1$ and thus, the allocation is strongly envy-free and Pareto optimal.
\end{proof}

\end{document}